\documentclass[11pt, a4paper]{article}
\usepackage[a4paper, left=30mm, right=30mm]{geometry} % margin=2.6cm

\usepackage[utf8]{inputenc}  % allow utf-8 input
\usepackage[T1]{fontenc}     % use 8-bit T1 fonts
\usepackage{ebgaramond}

\usepackage{hyperref}       % hyperlinks
\usepackage{url}            % simple URL typesetting
\usepackage{booktabs}       % professional-quality tables
\usepackage{amsfonts}       % blackboard math symbols
\usepackage{nicefrac}       % compact symbols for 1/2, etc.
\usepackage{microtype}      % microtypography
\usepackage{xcolor}         % colors
\usepackage{bbm}
\usepackage{amsthm}
\usepackage{rotating}

\usepackage{graphicx} % Required for inserting images
\usepackage[ngerman, english]{babel}
\usepackage[iso, ngerman]{isodate}

\usepackage{bbold}
\usepackage{mathtools}
\usepackage{amsmath} % Required for \DeclareMathOperator
\usepackage{nicefrac}
\usepackage{tikz}
\usepackage{subcaption}
\usepackage{centernot} % for the comparison

\newtheorem{theorem}{Theorem}

\usepackage{verbatim}

\RequirePackage[T1]{fontenc} 
\RequirePackage[tt=false, type1=true]{libertine} 
\RequirePackage[varqu]{zi4} 
\RequirePackage[libertine]{newtxmath}

\usepackage[round,comma]{natbib}
\title{From the Fair Distribution of Predictions to the Fair Distribution of Social Goods: Evaluating the Impact of Fair Machine Learning on Long-Term Unemployment\footnote{Sebastian Zezulka and Konstantin Genin. 2024. From the Fair Distribution of Predictions to the Fair Distribution of Social Goods: Evaluating the Impact of Fair Machine Learning on Long-Term Unemployment. In: Proceedings of the 2024 ACM Conference on Fairness, Accountability, and Transparency (FAccT '24). Association for Computing Machinery, New York, NY, USA, 1984–2006. \url{https://doi.org/10.1145/3630106.3659020}}}
\author{%
  Sebastian Zezulka
  \\
  University of Tübingen\\
  \texttt{sebastian.zezulka@uni-tuebingen.de} \\
   \and
  Konstantin Genin \\
  University of Tübingen \\
  \texttt{konstantin.genin@uni-tuebingen.de} \\
}
\date{June 2024}

\begin{document}

\maketitle

\begin{abstract}
Deploying an algorithmically informed policy is a significant intervention in society. Prominent methods for algorithmic fairness focus on the distribution of {\em predictions} at the time of {\em training}, rather than the distribution of {\em social goods} that arises {\em after} deploying the algorithm in a specific social context. However, requiring a `fair’ distribution of predictions may undermine efforts at establishing a fair distribution of social goods. First, we argue that addressing this problem requires a notion of {\em prospective fairness} that anticipates the change in the distribution of social goods {\em after} deployment. Second, we provide formal conditions under which this change is identified from pre-deployment data. That requires accounting for different kinds of performative effects. Here, we focus on the way predictions change policy decisions and, consequently, the causally downstream distribution of social goods. Throughout, we are guided by an application from public administration: the use of algorithms to predict who among the recently unemployed will remain unemployed in the long term and to target them with labor market programs. Third, using administrative data from the Swiss public employment service, we simulate how such algorithmically informed policies would affect gender inequalities in long-term unemployment. When risk predictions are required to be `fair' according to statistical parity and equality of opportunity, targeting decisions are less effective, undermining efforts to both lower overall levels of long-term unemployment and to close the gender gap in long-term unemployment. 
\end{abstract}

\section{A fundamental question for fair machine learning}\label{sec:fundamental}
Research in algorithmic fairness is often motivated by the concerns that machine learning algorithms will reproduce or even exacerbate structural inequalities reflected in their training data \citep{lum2016predict, Tolbert2023}. Indeed, whether an algorithm exacerbates an existing social inequality is emerging as a central compliance criterion in EU non-discrimination law \citep{Weerts2023_CONF}. However, many methodological solutions developed by researchers in algorithmic fairness are, surprisingly, ill-suited for addressing this fundamental question. At some level, the questions of algorithmic fairness are ill-posed: often, it does not make sense to talk about the fairness of a predictor independent of the policy context in which it is deployed. It is our policies and their effects that are just or unjust; `fair' predictors can both support unjust policies and undermine just policy. For example, public employment services use predictions of the risk of long-term unemployment (LTU) to decide who is given access to labor market programs. Policy doves target those at the highest risk with training programs, while hawks, considering those at the highest risk to be hopeless cases, withhold training on grounds of `efficiency'. The social consequences of prediction errors differ significantly depending on how these predictions will be used. It would be surprising if we could say whether a predictor is fair independent of this policy context. Therefore, rather than focusing on the distribution of {\em predictions} at the time of {\em training}, we focus on the distribution of {\em social goods} induced by {\em deploying} a predictive algorithm in a policy context. Our point is not that formal fairness constraints on predictions would always make things worse, but rather that part of due diligence is forecasting their effects on outcomes. In our case study, we focus on the gender gap in long-term unemployment as one such outcome.

The field of algorithmic fairness has produced many mathematical demonstrations of necessary trade-offs between different notions of `fairness', and between 'fair' and accurate prediction \citep{Borsboom2008, Kleinberg2016, Chouldechova2017, Menon2018}. This lends the field an air of tragedy and makes the pursuit of fairness seem fundamentally quixotic. But, while mathematical trade-offs exist between predictive accuracy and the `fair' distribution of predictions, predictive accuracy does not necessarily trade-off against the fair distribution of social goods \citep{Tal2023_CONF, CorbettDavies2018, Green2022}. Indeed, we should expect that accurate predictions help us to effectively implement policy aimed at ameliorating unjust inequalities. In our empirical case study, we demonstrate that (1) requiring risk predictions to be `fair' in terms of statistical parity and equal opportunity \textit{undermines} efforts to lower overall levels of long-term unemployment and to close the gender gap in long-term unemployment, (2) that the hawkish policy of withholding training programs from those at the highest risk is \textit{no more efficient} than the dovish policy of prioritizing those with the highest risk, and (3) that accurate prediction of {\em counterfactual} treatment outcomes, rather than risk scores, enables individualized targeting and therefore, a better and more equitable distribution of social goods. 

Of course, this shift in focus poses methodological challenges. To anticipate the causal effects of embedding a predictive algorithm into a social process, we must make some effort to, first, identify the contextually relevant inequalities in the distribution of social goods; second, understand the policy processes and decisions that partially give rise to, and could conceivably ameliorate these inequalities; and third, model how algorithmic predictions might {\em change} these processes and, therefore, the distribution of social goods. Standard algorithmic fairness methods neglect every part of this process \citep{Green2022, Selbst2019_CONF}. All of these methods impose some constraints on predictions that hold in the (retrospective) training distribution. By focusing on the distribution of predictions at the time of training, they obscure substantive inequalities in real-world quantities and neglect the changes in decision-making that arise from the deployment of predictive algorithms. Consequently, these methods fail to anticipate the effects of {\em deploying} these algorithms on the distribution of social goods. Here, we address these shortcomings in the following way:
\begin{itemize}
    \item We reconceptualize algorithmic fairness questions as policy problems: {\em Prospective fairness} requires efforts to anticipate the impact of deploying an algorithmically informed policy on inequality in social goods.
    \item We state formal conditions under which the effect of deploying an algorithmically informed policy on context-relevant inequalities is identified from pre-deployment data.
    \item We illustrate our approach with an extensive case study on the statistical profiling of registered unemployed using a rich administrative dataset from Switzerland. We study the likely effects of two algorithmic policy proposals on the gender gap in the rates of long-term unemployment.
\end{itemize}

Our case study is based on administrative data from the Swiss Active Labor Market Policy Evaluation Dataset \citep{Lechner2020}. The original sample, collected in 2003, contains roughly one hundred thousand observations of registered unemployed aged $24$ to $55$. Although most unemployed were not assigned to any program, we observe outcomes for six labor market programs. The Swiss labor market, as outlined in Section~\ref{sec:statprofiling}, is characterized by an overall unemployment rate of about $4\%$, a high rate of long-term unemployment (LTU), and a persistent gender reemployment gap (Figure~\ref{fig:LTUGapSwiss}). In the administrative data, the LTU gender gap is at $3.9\%$, with an LTU rate of $43.6\%$ among women and $39.7\%$ among men. The gap between Swiss citizens and non-citizens is at $15.8\%$, with a rate of $35.7\%$ among Swiss citizens and $51.5\%$ among non-citizens.

The plan of the paper is as follows: first, we argue for {\em prospective fairness} as a conceptual framework and survey related work; section~\ref{sec:statprofiling} introduces two recently proposed algorithmic policies intended to support public employment agencies in reducing long-term unemployment; we argue that, in this context, the gender gap in long-term unemployment is a simple and intuitive measure of systemic inequality; section~\ref{sec:Theorem} formalizes conditions under which the causal effect of deploying an algorithmically informed policy on a measure of systematic inequality is identified from pre-deployment data; in section~\ref{sec:CaseStudy} we illustrate the method with an extended case study, simulating two proposed profiling policies and their effects on the gender reemployment gap. Section~\ref{sec:Conclusion} concludes and outlines directions for future work.

%%%%%%%%%%%%%%%%%%%%%%%%%%%%%%%%%%%%%%%%%%%%%%%%%%%%%%%%%%%%%%%%%%%%%%%%%%%%%%%%%%%%
\section{From retrospective to prospective fairness}\label{sec:prospectiveFairness}
In paradigmatic risk-assessment applications, machine learners are concerned with learning a function that takes as input some features $X$ and a sensitive attribute $A$ and outputs a score $R$ which is valuable for predicting an outcome $Y$. The algorithmic score $R$ is meant to inform some important decision $D$ that, typically, is causally relevant for the outcome $Y$. In the application that concerns us in this paper, features such as the education and employment history $(X)$ and gender $(A)$ of a recently unemployed person are used to compute a risk score $(R)$ of long-term unemployment $(Y).$ This risk score $R$ is meant to support a caseworker at a public employment agency in making a plan $(D)$ about how to re-enter employment. This plan may be as simple as requiring the client to apply to some minimum number of jobs every month or referring them to one of a variety of job-training programs. 

Formal fairness proposals require that some property is satisfied by either the joint distribution $P(A,X,R,D,Y)$ or the causal structure $G$ giving rise to it. Individual fairness proposals introduce a similarity metric $M$ on $(A,X)$ and suggest that similar individuals should have similar risk scores. In all these cases, the relevant fairness property is a function $\varphi(P,G,M)$. Group-based fairness \citep{barocas-hardt-narayanan} ignores all but the first parameter; causal fairness \citep{Kilbertus2017, Kusner2017_CONF} ignores the last; and individual fairness  \citep{Dwork2012_CONF} ignores the second. All these proposals agree that fairness is a function of the distribution (and perhaps the causal structure) at the time when the prediction algorithm has been trained, {\em but before it has been deployed}. We claim that addressing this fundamental question of fair machine learning requires comparing the status quo {\em before} deployment with the situation likely to arise {\em after} deployment. In other words: {\em prospective} fairness is a matter of anticipating the change from $\varphi(P_{\text{pre}}, D_{\text{pre}}, M)$ to $\varphi(P_{\text{post}},D_{\text{post}}, M)$. We do not claim that there is a single correct inequality measure $\varphi(\cdot),$ nor even that there is an all-things-considered way of trading off different candidates, only that we must make a good faith effort to anticipate changes in the relevant measures of inequality. 

As shown in Figure~\ref{fig:DAG_D0_D1}, deploying a decision support algorithm introduces a causal path from the predicted risk scores $R$ to the decisions $D$. Importantly, the outcome variable $Y$ is causally downstream of this intervention. The addition of a causal path can be modeled as a {\em structural} intervention \citep{malinsky2018intervening, Bynum2023}.

\begin{figure}[hbt!]
\centering
  \begin{subfigure}{0.45\textwidth}
    \centering
    \includegraphics[width=0.61\linewidth]{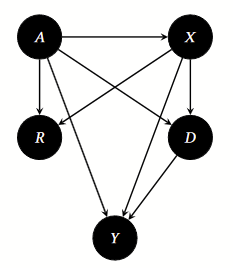}
    \caption{Causal structure $G_{\text{pre}}$ \textit{before} deploying an algorithmically informed policy.}
    \label{fig:sub1}
  \end{subfigure}
  \hfill
  \begin{subfigure}{0.45\textwidth}
    \centering
    \includegraphics[width=0.61\linewidth]{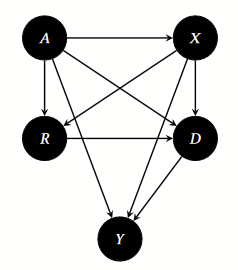}
    \caption{Causal structure $G_{\text{post}}$ \textit{after} deploying an algorithmically informed policy.}
    \label{fig:sub2}
  \end{subfigure}
\caption{The left hand side shows the pre-deployment causal graph $G_{\text{pre}}$ inducing a joint probability distribution $P_{\text{pre}}$ over sensitive attributes $A$, features $X$, risk score $R$, decision $D$, and outcome variable $Y$. The risk score $R$ is the output of a learned function from $A$ and $X$. Since this graph represents the situation after training, but before deployment, there is no arrow from the risk score $R$ to the decision $D$. \textit{Retrospective} fairness formulates constraints $\varphi(G_{\text{pre}}, P_{\text{pre}}, M)$ on the pre-deployment arrangement alone. The right-hand side represents the situation after the algorithmically informed policy has been deployed, with predictions $R$ now affecting decisions $D$. Prospective fairness requires comparing the consequences of intervening on the structure of $G_{\text{pre}}$ and moving to $G_{\text{post}}$. In other words, comparing $\varphi(G_{\text{pre}}, P_{\text{pre}}, M)$ with $\varphi(G_{\text{post}}, P_{\text{post}}, M)$.}
\label{fig:DAG_D0_D1}
\end{figure}

From a dynamical perspective, static and retrospective fairness proposals go wrong in two ways. In the worst case, they are {\em self-undermining}. \citet{Mishler2022} show that meeting the fairness notions of sufficiency $(Y\perp A~|~R)$ or separation $(R\perp A~|~Y)$ at the time of training necessitates that they will be violated after deployment. In terms of sufficiency, where $\perp$ denotes (conditional) statistical independence, we have that:  
\begin{displaymath}
    Y\perp_{\text{pre}} A~|~R \hspace{4pt}\text{ entails } \hspace{4pt}  
Y \not{\perp}_{\text{post}} A~|~R. 
\end{displaymath}
Group-based notions of fairness like sufficiency and separation that feature the outcome fall victim to {\em performativity}: the tendency of an algorithmic policy intervention to shift the distribution away from the one on which it was trained \citep{Perdomo2020}. But as \citet{Mishler2022} show, they are undermined not by an unintended and unforeseen performative effect, but by the {\em intended, and foreseen} shift in distribution induced by algorithmic support, i.e.:
\begin{displaymath}
    P_{\text{pre}}(D~|~A,X,R)\neq P_{\text{post}}(D~|~A,X,R).
\end{displaymath}
In other words, they are undermined by the fact that algorithmic support changes decision-making, which, presumably, is the point of algorithmic support in the first place. Since the distribution of the outcome $Y$ will change after deployment, \citet{Berk2023a} advises against group-based metrics involving it, opting for statistical parity ($R\perp A$) instead. 

It is not likely that individual and causal fairness proposals are so drastically self-undermining. So long as the similarity metric stays constant, an algorithm that treats similar people similarly will continue to do so after deployment. If, as \citet{Kilbertus2017} suggest, causal fairness is a matter of making sure that all paths from the sensitive attribute $A$ to the prediction $R$ are appropriately mediated, then causal fairness is safe from performative effects so long as the qualitative causal structure \textit{upstream} of the prediction $R$ remains constant.  

But even if causal and individual fairness proposals are not so dramatically self-undermining, they are simply {\em not probative} of whether the algorithm reproduces or exacerbates inequalities in the distribution of social goods, since these are causally \textit{downstream} of algorithmic predictions. In particular, it is customary to ignore the real-world dependence between $A$ and $Y$ induced by the social status quo as the target of an intervention, since nothing can be done about it at the time of training. Instead, fairness researchers focused on whether the risk score \textit{itself} is fair, whether in the group, individual, or causal sense. However, from the dynamical perspective, it is perfectly reasonable to ask whether the proposed algorithmic policy will exacerbate the systemic inequality reflected in the dependence between gender $(A)$ and long-term unemployment $(Y)$. Indeed, simple dynamical models and simulations suggest that algorithms meeting static fairness notions at training may in the long run exacerbate inequalities in outcomes \citep{Liu2018_CONF,Zhang2021_CONF}.
We derive formal conditions under which the effect of deploying an algorithmic policy on the joint distribution of $(Y,A)$ is identified from pre-deployment data and provide a realistic case study analyzing the effects of algorithmic policies in public employment on the gender gap in long-term unemployment. 

\subsection{Related Work}\label{sec:relatedwork}
In machine learning, the fairness debate began with risk assessment tools for decision- and policy-making \citep{Angwin2016, Kleinberg2016, Chouldechova2017, Mitchell2021}. To this day, many standard case studies e.g., lending, school admissions, and pretrial detention, fall within this scope. See \citet{Berk2023} for a review on fairness in risk assessment and \citet{Borsboom2008} and \citet{Hutchinson2019_CONF} for predecessors in psychometrics. Since then, researchers have stressed the importance of explicitly differentiating policy decisions from the risk predictions that inform them \citep{Barabas2018_CONF, Mouzannar2019_CONF, Kuppler2022, Beigang2022, Tal2023_CONF, Liu2024} and of studying machine learning algorithms in their socio-technological contexts \citep{Selbst2019_CONF, Grote2023}. We incorporate both of these insights into the present work.

A central negative result emerging from recent fairness literature highlights the dynamically self-undermining nature of group-based fairness constraints that include the outcome variable $Y$. \citet{Mishler2022} show that a classifier that is formally fair in the training distribution will violate the respective fairness constraint in the post-deployment distribution. \citet{Coston2020_CONF} suggests that the group-based fairness notion be formulated instead in terms of the potential outcomes $Y^d.$ These alternative proposals are no longer self-undermining, but they are still not testing the policy's effect on inequality in the distribution of social goods. This paper builds upon the negative results of \citet{Berk2023a} and \citet{Mishler2022}: we show how the post-interventional effect of an algorithmically informed policy on the distribution of social goods can be identified from a combination of (1) observational, pre-deployment data and (2) models of the policy proposal. 

An emerging literature on long-term fairness focuses on the dynamic evolution of systems under sequential-decision making, static fairness constraints, and feedback loops; see \citet{Zhang2021_CONF} for a survey. \citet{Ensign2017} consider predictive feedback loops from selective data collection in predictive policing. \citet{Hu2018_CONF} propose short-term interventions in the labor market to achieve long-term objectives. Using two-stage models, \citet{Liu2018_CONF} and \citet{Kannan2019_CONF} show that retrospective fairness constraints can, under some conditions, have negative effects on outcomes in disadvantaged groups. With simulation studies, \citet{DAmour2020_CONF} and \citet{Zhang2020} confirm that imposing static fairness constraints does not guarantee that these constraints are met over time and can, under some conditions, exacerbate inequalities in social goods. \citet{Scher2023} model long-term effects of statistical profiling for the allocation of unemployed into labor market programs on skill levels. The picture emerging from this literature is that post-interventional outcomes of algorithmic policies are a relevant dimension for normative analysis that is not adequately captured by retrospective fairness notions designed to hold in the training distribution. 

\section{Statistical profiling of the unemployed}\label{sec:statprofiling}
Since the 1990s, participation in active labor market programs (ALMPs) has been a condition for receiving unemployment benefits in many OECD countries \citep{Considine2017}. ALMPs take many forms, but paradigmatic examples include resume workshops, job-training programs, and placement services (see \citet{Bonoli2010} for a helpful taxonomy). Evaluations of ALMPs across OECD countries find small but positive effects on labor market outcomes \citep{Card2018, Vooren2018, Lammers2019}. Importantly, the literature also reports large effect-size heterogeneity between programs and demographics, as well as assignment strategies that are as good as random for Switzerland \citep{Knaus2022a}, Belgium \citep{Cockx2023}, and Germany \citep{Goller2021}. This implies potential welfare gains from a more targeted allocation into programs, especially when taking into account opportunity costs---a compelling motivation for algorithmic support. Indeed, the subsequent case study suggests that, if allocation decisions are made based on data-driven estimates of individualized treatment effects, the gender reemployment gap, as well as overall long-term unemployment, can be significantly reduced.

Statistical profiling of the unemployed is current practice in various OECD countries including Australia, the Netherlands, and Flanders, Belgium \citep{Desiere2019}. Paradigmatically, supervised learning techniques are employed to predict who is at risk of becoming long-term unemployed (LTU) \citep{Mueller2023_TECH_REPORT}. Such tools are regularly framed as introducing objectivity and effectiveness in the provision of public goods and align with demands for evidence-based policy and digitization in public administration. ALMPs target \textit{supply-side} problems by increasing human capital and \textit{matching} problems by supporting job search. \textit{Demand-side} policies that focus on the creation of jobs are not considered \citep{Green2022}.

Individual scores predicting the risk of long-term unemployment support a variety of decisions. For example, the public employment service (PES) of Flanders so far uses risk scores only to help caseworkers and line managers decide who to contact first, prioritizing those at higher risk \citep{Desiere2020}. In contrast, the PES of Austria (plans to) use risk scores to classify the recent unemployed into three groups: those with good prospects in the next six months; those with bad prospects in the next two years; and everyone else. The proposed policy of the Austrian PES is to focus support measures on the third group while offering only limited support to the other two. Advocates claim that, since ALMPs are expensive and would not significantly improve the re-employment probabilities of individuals with very good or very bad prospects, considerations of cost-effectiveness require a focus on those with middling prospects \citep{Allhutter2020}. However intuitive this may seem, it is nowhere substantively argued that statistical predictions of long-term unemployment from observational data can be reliably used as estimates for the effectiveness of administrative interventions. One worry is that the unemployed who are labeled high-risk may be similar to those who, historically, received ineffective programs. This is further complicated by the presence of long-standing structural inequalities in the labor market, which may be reproduced by algorithmic policies leaving those with ``poor prospects" to their own devices. In the subsequent simulation study, the efficiency claims made in favor of Austrian-style policy are not corroborated.

\begin{figure}[hbt!]
  \centering
  \includegraphics[width=\linewidth]{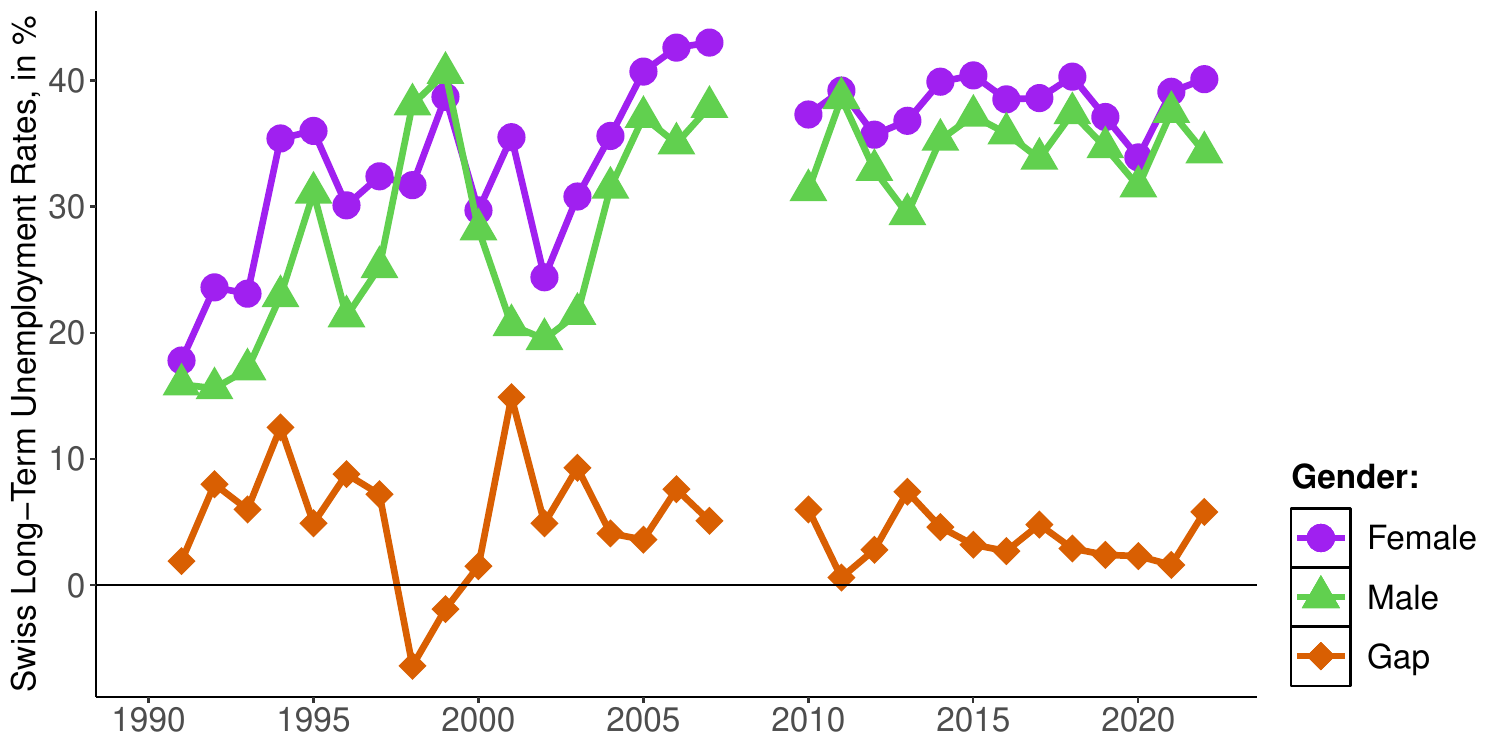}

\caption{Swiss Long-Term Unemployment Rates by Gender. Data for the period 2010-2022 are from Eurostat \citep{eurostat_swissLTUGap}; the gender rates in long-term unemployment are computed as the share of all unemployed men/women aged 20-64 who are unemployed for more than a year. Data for the period 1991-2007 are from the 2012 Swiss Social Report \citep{buhlmann2013swiss}, where age information is not available. Data for 2008-9 is not readily available.}
\label{fig:LTUGapSwiss}
\end{figure}

Labor markets in OECD countries are structured by various inequalities. Gender is a particularly long-standing and significant axis of inequality in labor markets, with the gender pay gap and the child penalty being notorious examples \citep{Kleven2023, Bishu2016}. On the other hand, the gender gap in unemployment rates has largely disappeared over the last decades \citep{Albanesi2018}. Nevertheless, structural differences in unemployment dynamics remain. For example, although women in Germany are less likely to enter into unemployment, their exit probabilities are also lower \citep{arbeitsmarkt2022statistik}. Similarly, there is a longstanding gender gap in long-term unemployment in Switzerland (see Figure \ref{fig:LTUGapSwiss}). The obvious worry is that prediction algorithms will pick up on these historical trends, as demonstrated in \citet{Kern2021}. The Austrian proposal for an LTU prediction algorithm furnishes a particularly dramatic example. That algorithm takes as input an explicitly gendered feature ``obligation to care'', which has a negative effect on the predicted re-employment probability and, by design, is only active for women \citep{Allhutter2020}. This controversial design choice was justified as reflecting the ``harsh reality" of the gendered distribution of care responsibilities. Whatever the wisdom of this particular variable definition, many other algorithms would pick up on the same historical patterns. Moreover, if the intended use of these predictions is to withhold support for individuals at high risk of long-term unemployment, it is clear that such a policy might exacerbate the situation by further punishing women for greater care obligations. 

The preceding underscores the need for a {\em prospective} fairness methodology that assesses whether women's actual re-employment probability suffers under a proposed algorithmic policy. More abstractly, what is needed is a way to predict how the pre-deployment probability $P_\text{pre}(Y~|~A)$ will compare with the post-deployment probability $P_\text{post}(Y~|~A)$. With these estimates in hand, it would also be possible to predict whether the gender gap in long-term unemployment is exacerbated, or ameliorated, under a proposed algorithmic policy. This gender gap is one particular choice for a fairness notion $\varphi(\cdot).$ Variations on this simple metric could be relevant in many other settings. For example, gender gaps in hiring, or racial disparities in incarceration could be criteria that an algorithmically informed policy should, minimally, not exacerbate \citep{Kasy2021_CONF}. In the following section, we give general conditions under which the post-deployment change in the joint distribution of the outcome $(Y)$ and the sensitive attribute $(A)$ is identified from pre-deployment data.

%%%%%%%%%%%%%%%%%%%%%%%%%%%%%%%%%%%%%%%%%%%%%%%%%%%%%%%%%%%%%%%%%%%%%%%%%%%%%%%%%%%%
\section{Identifiability of the Post-Deployment Distribution}\label{sec:Theorem}
Let $A,X,R,D,Y$ be discrete, {\em observed} random variables. Here, $A$ represents gender; $X$ represents baseline covariates observed by the public employment service; $R$ is an estimated risk of becoming long-term unemployed; $D$ is an allocation decision made by the public employment service and $Y$ is a binary random variable that is equal to $1$ if an individual becomes long-term unemployed. For simplicity, we assume that $R$ is a deterministic function of $A$ and $X$. We write $\mathcal{A,X,R,D,Y}$ for the respective ranges of these random variables. For $d\in\mathcal{D},$ let $Y^d$ be the potential outcome under policy $d,$ in other words: $Y^d$ represents what the long-term unemployment status of an individual {\em would have been} if they had received allocation decision $d.$ Naturally, $Y^1,\ldots, Y^{|\mathcal D|}$ are not all observed. Our first assumption is a rather mild one; we require that the observed outcome for individuals allocated to $d$ is precisely $Y^d:$ 
\begin{equation} 
    \tag{\textsc{Consistency}} 
    Y = \sum_{d\in\mathcal{D}} Y^d\mathbbm{1}[D=d]. 
\end{equation}
Consistency is to be interpreted as holding both before and after the algorithmic policy is implemented.

More substantially, we assume that the potential outcomes and decisions are unconfounded given the observed features $(A, X)$ both before and after the intervention:
\begin{equation}
    \tag{\textsc{Unconfoundedness}}
    Y^d \perp_t D~|~A, X.
\end{equation}
Unconfoundedness is a rather strong assumption that requires that the observed features $A, X$ include all common causes of the decision and outcome. In the case of a fully automated algorithmic policy, unconfoundedness holds by design; but usually, risk assessment tools are employed to support human decisions, not fully automate them \citep{Levy2021}. Although it is not fated that all factors relevant to a human decision are available to the data analyst, unconfoundedness is reasonable if rich administrative data sets capture most of the information relevant to allocation decisions. For a case in which this assumption fails, see \citet{Petersen2021}.

We have argued that, to address our fundamental question of fair machine learning, one must predict whether implementing the candidate algorithmically informed policy leads to an improvement, or at least no deterioration, in the distribution of social goods. In the running example, this amounts to comparing features of $P_{\text{pre}}(Y~|~A)$ with $P_{\text{post}}(Y~|~A)$. The first distribution is trivial to estimate, but how to estimate $P_{\text{post}}(Y~|~A)$ from pre-deployment data? Here, the fundamental problem is performativity \citep{Perdomo2020}. Our policy intervention will, in all likelihood, change the process of allocation into labor market programs and, thus, change the distribution of outcomes we are interested in. But not all kinds of performativity are equal. Some performative effects are intended and foreseeable. For example, the {\em algorithmic} effect is the intended change in decision-making due to algorithmic support:
\begin{equation*}
    P_{\text{pre}}\left(D=d~|~A=a, X=x\right) \neq P_{\text{post}}\left(D=d~|~A=a, X=x\right). \tag{\textsc{Algorithmic Effect}}
\end{equation*}
The first term in this inequality is the propensity score which can be directly estimated from training data. The second term cannot be directly estimated {\em ex-ante}. Nevertheless, it is possible to make reasonable conjectures about the second term given a concrete proposal for how risk scores should inform decisions. For example, if $D$ is binary, we could model the Austrian proposal as providing support so long as the risk score is neither too high nor low:
    \[ P_{\text{post}}(D=1~|~A=a, X=x) = \mathbbm{1} \left[l<R(a,x)<h\right].\]
More complex proposals for how risk scores should influence decisions require more careful modeling. The subsequent empirical case study delivers a more realistic model. 

Although we allow for algorithmic effects, these cannot be too strong---the policy cannot create allocation options that did not exist before. That is, the risk assessment tools only change allocation probabilities into {\em existing} programs. Moreover, we assume that the policy creates no unprecedented allocation-demographic combinations: 
\begin{equation}
    \tag{\textsc{No Unprecedented Decisions}}
    P_{\text{pre}}(D=d~|~A=a, X=x)>0 \text{ if } P_{\text{post}}(D=d~|~A=a, X=x)>0.
\end{equation}
This would be violated if e.g., no women were allocated to some program before the policy change.

Throughout this paper, we assume that no other forms of performativity occur. In particular, we assume that the conditional average treatment effects (CATEs) of the allocation on the outcome are stable across time:
\begin{equation}
    \tag{\textsc{Stable CATE}} 
    P_{\text{pre}}\left(Y^d~|~A=a, X=x\right) = P_{\text{post}}\left(Y^d~|~A=a, X=x\right).
\end{equation}
This amounts to assuming that the effectiveness of the programs (for people with $A=a, X=x$) does not change, so long as all that has changed is the way we {\em allocate} people to programs. In the case study, we assume that conditional average treatment effects are stable under changes to allocation policies, as well as to the total number of places available in (capacities of) each program. This assumption could be violated if e.g., a program works primarily by making some better off only at the expense of others---if everyone were to receive such a program, it would have no effect \citep{Crepon2013}.

While \textit{algorithmic effects} of deployment are intended and, to some degree, foreseeable types of performativity, \textit{feedback} effects that change the covariates are more complicated to model.\footnote{In the classification of \citet{Pagan2023_CONF}, we focus on what they call ``Outcome Feedback Loops''. In our terminology, performativity is not exhausted by feedback effects.} Following \citet{Mishler2022} and \citet{Coston2020_CONF}, we assume away the possibility of feedback effects, leaving these for future research:
\begin{equation}
    \tag{\textsc{No Feedback}}
    P_{\text{pre}}\left(A=a, X=x\right) = P_{\text{post}}\left(A=a, X=x\right).
\end{equation}
 \textsc{No Feedback} amounts to assuming that the baseline covariates of the recently employed are identically distributed pre- and post-deployment. Strictly speaking, this is false, since the decisions of caseworkers will affect the covariates of those who re-enter employment and some of them will, eventually, become unemployed again. However, since the pool of employed is much larger than the pool of unemployed, the policies of the employment service have much larger effects on the latter than the former. For this reason, we may hope that feedback effects are not too significant.

\textsc{No Unprecedented Decisions, Stable CATE and No Feedback} might fail dramatically if e.g., the deployment of the policy coincided with a major economic downturn. In a serious downturn, the employment service may have to assist people from previously stable industries (violating \textsc{No Unprecedented Decisions} and \textsc{No Feedback}), or employment prospects might deteriorate for everyone (violating \textsc{Stable CATE}). However, the possibility of such exogenous shocks is not a threat to our methodology. We are interested in the {\em ceteris paribus} effect of the algorithmic policy on structural inequality, not an all-thing-considered prediction of future economic conditions.

We are now in a position to show that, under the assumptions outlined above, it is possible to predict $P_{\text{post}}(Y=y~|~A=a)$ from pre-interventional data and a supposition about $P_{\text{post}}(D=d~|~A=a,X=x)$. That means that we can also predict changes to the overall reemployment probability $P_\text{post}(Y=0)$ as well as the gender reemployment gap $P_\text{post}(Y=1~|~A=1)-P_\text{post}(Y=1~|~A=0).$ Each of these are natural and important instances of $\varphi(\cdot).$ The proof is deferred to the supplementary material.

\begin{theorem}\label{thm:Identification}
    Suppose that \textsc{Consistency, Unconfoundedness, No Unprecedented Decisions, Stable CATE} and \textsc{No Feedback} hold. Suppose also that $P_{\text{post}}(A=a)>0$. Then, $P_{\text{post}}(Y=y~|~A=a)$ is given by
    \[ \sum_{(x,d)\in \Pi_{\text{post}}} P_{\text{pre}} (Y=y~|~A=a,X=x,D=d)P_{\text{pre}} (X=x~|~A=a)P_{\text{post}}(D=d~|~A=a,X=x),\]
where $\Pi_t = \left\{ (x,d) \in \mathcal{X \times D} : P_{t}(X=x,D=d~|~A=a)>0\right\}.$
\end{theorem}

Note that the first two terms in the product are identified from pre-deployment data. Given a sufficiently precise proposal for how risk scores influence decisions, it is also possible to model $\Pi_{\text{post}}$ and the last term before deployment. This allows us to systematically compare different (fairness-constrained) algorithms and decision procedures, and arrive at a reasonable prediction of their combined effect on reemployment probabilities (and the gender reemployment gap) before they are deployed. In the following, we show how this approach works in a realistic case study.

%%%%%%%%%%%%%%%%%%%%%%%%%%%%%%%%%%%%%%%%%%%%%%%%%%%%%%%%%%%%%%%%%%%%%%%%%%%%%%%%%%%%
\section{Long-term Unemployment in Switzerland}\label{sec:CaseStudy}
Prospective fairness requires forecasting the effect of using (fair) risk scores to inform program allocation decisions on both the overall risk of long-term unemployment and the gender gap in long-term unemployment. We present an extensive case study based on Swiss administrative data to study three questions: do fairness-constrained risk scores improve outcomes? are restrictive, Austrian-style allocation policies more efficient than Flemish-style policies that prioritize people at high risk? and can we improve outcomes with individualized estimates of program effectiveness?

\subsection{Methodology}
Our analysis proceeds in the following stages: (1) Using double-robust machine learning, we first estimate the effectiveness of each of the programs for all individuals in our test sample. (2) We estimate risk scores for the individuals in our test sample, using fairness-constrained and fairness-unconstrained methods. We implement two fairness constraints: statistical parity and equal opportunity. (3) For each of the risk scores from stage two, we prioritize the individuals in the test sample. The Flanders-style policy prioritizes those at the highest risk. The Austrian prioritization does the same, but only for those in the $30-70$th risk percentiles; the rest go to the end of the line. (4) For each priority list from stage three, we assign unemployed to programs until program capacity is reached. We model two assignment schemes. The first assigns individuals to programs randomly. The second uses the results of stage one to assign individuals to the program with the highest estimated effectiveness. Additionally, we consider the effect of increasing program capacities. Finally, we summarize the effects of different combinations of choices from steps (2-4) on overall rates of long-term unemployment and the gender-reemployment gap.\footnote{The replication package for this analysis is available on Github: \url{https://github.com/sezezulka/2023-01-ALMP-LTU.git}.}

\subsubsection{Data}
We exploit the administrative Swiss Active Labor Market Policy (ALMP) Evaluation Dataset.\footnote{The data is available for scientific use at SWISSbase \citep{Lechner2020}.} The original sample contains observations on $100,120$ registered unemployed in $2003$, aged $24$ to $55$. Recently unemployed received one of seven treatments: \textit{no program, vocational training, computer programs, language courses, job search programs, employment programs, and personality training}. Among the seven treatment options, \textit{no program} and \textit{job search programs} are by far the most common treatments. We restrict the analysis to the German-speaking cantons as assignment strategies differ among the three language regions \citep{Knaus2022}. To avoid overstating the effectiveness of ``no program'', we estimate pseudo program starting points for individuals in this treatment arm and exclude those who are re-employed before the pseudo starting point \citep{Lechner1999, Knaus2022}. This results in the exclusion of $5,076$ observations.\footnote{The problem is that some people are {\em assigned} to ``no program'' while others exit unemployment before they can receive an assignment but these are coded the same way. Compare: if someone spontaneously recovers before being assigned to an arm of a drug trial, this should not count in favor of the placebo.} 

The final data set contains $64,296$ individuals, which we divide equally into training and test sets. The simulation study is performed on the test set of $32,148$ individuals and all results are reported for this population. Descriptive statistics for the simulation data are reported in Table~\ref{tab:descriptives} in the Appendix. 

For all individuals, we observe employment status for $36$ months after registration with the Swiss Public Employment Service (PES). Our target, long-term unemployment, is defined as a binary variable indicating continuous unemployment for $12$ months after the (pseudo) program start.\footnote{This is a deviation from \citet{Koertner2023}, who define their target variable as $12$ months after registration with the PES.} The treatment variable is defined as the first program assigned within six months after registering as unemployed. The administrative data includes information on the individual employment biographies, demographics, and local labor market conditions as well as information on the individual caseworker and their assessment of their clients' labor market outlook. 

\subsubsection{Individualized Average Potential Outcomes}
We adopt double-robust machine learning for the estimation of individual average potential outcomes (IAPOs) and treatment effects (IATEs) for the seven treatment options  \citep{Chernozhukov2018, Alaa2023, Curth2024}. We follow \citet{Knaus2022} and \citet{Koertner2023} in their identification strategy and use the R-package \textsc{causalDML} \citep{Knaus2022}. Inverse probability weighting is used to account for non-random selection into the programs under the identifying assumptions of {\em Unconfoundedness} (similar to our \textsc{Unconfoundedness}), {\em Common Support} (\textsc{No Unprecedented Decisions}), and {\em Stable Unit Treatment Value} (\textsc{Consistency} and \textsc{Stable CATE}). Especially important for the plausibility of Unconfoundedness is the availability of information about the individual caseworker. See Appendix~\ref{supp:iapos} for a more detailed discussion of the estimation approach.

The resulting (individualized) average treatment effects are given in Figure~\ref{fig:IATEs}. They are in line with the results reported in \citet{Knaus2022} and \citet{Koertner2023a}. Vocational Training, Computer Programs, and Language Courses have the strongest effects on reducing (long-term) unemployment. We find that Job Search and Employment Programs on average increase the risk of long-term unemployment by between $2$ to $3$ percentage points and confirm the high effect heterogeneity in all treatments. The reported treatment effects are the difference of the respective potential outcome scores, where ``no program'' is the baseline program. IATEs broken down by gender are given in Figure~\ref{fig:IATEs-by-Gender}.

\begin{figure}[hbt!]
  \begin{subfigure}{0.49\textwidth}
    \centering
    \includegraphics[width=\linewidth]{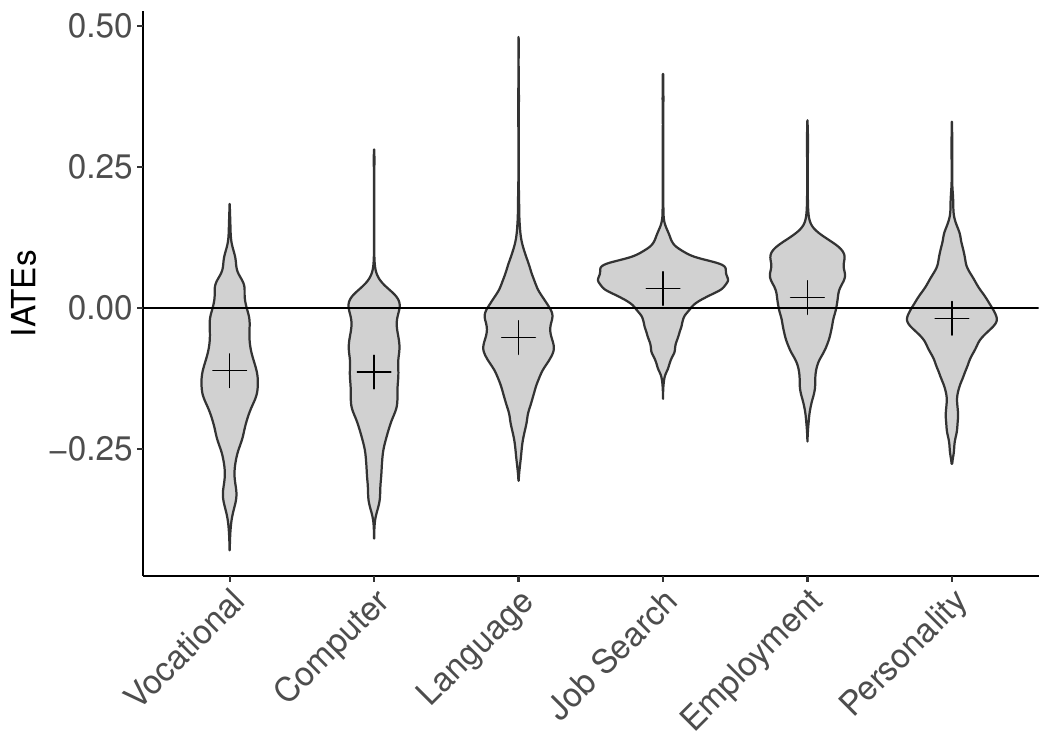}
    \caption{Individualized Average Treatment Effects.}
  \end{subfigure}%
\hfill
  \begin{subfigure}{0.5\textwidth}
    \centering
    \begin{tabular}{lccc}
    \toprule
    & ATE & SE & 95\%-CI  \\
    \midrule
    \midrule
    Vocational & -11.12 & 0.06 & [-11.12, -11.12] \\
    Computer     & -11.37 & 0.05 & [-11.37, -11.37] \\
    Language & -5.25 & 0.04 & [-5.26, -5.25]\\
    Job Search & 3.43 & 0.03 & [3.43, 3.43] \\
    Employment & 1.83 & 0.04 & [1.83, 1.83] \\
    Personality & -1.84 & 0.04 & [-1.84, -1.84] \\
    \bottomrule
    \end{tabular}
    \caption{Average Treatment Effects in percentage points, standard errors, and $95\%$ confidence intervals. Negative treatment effects imply a lower risk of becoming long-term unemployed.}
    \label{tab:iates_by_program}
  \end{subfigure}
  \caption{Estimated (Individualized) Average Treatment Effects for six labor market programs with ``no program'' as the baseline.} 
  \label{fig:IATEs}
\end{figure}

\subsubsection{Risk scores}
In 2003, program assignment in the Swiss public employment service was made at the discretion of the individual caseworker. This practice continues to this day.\footnote{The canton of Freiburg had a pilot study from 2012-2014, providing caseworkers with estimates of the expected length of the unemployment spell \citep{arni2016}.}  For estimating the risk scores to determine the prioritization, all caseworker information is excluded so that only data reasonably available at registration time is used. The sensitive attributes are included and the full list of features is given in Appendix~\ref{supp:risk}. 

We estimate fairness-unconstrained risk scores as well as risk scores constrained to satisfy statistical parity\footnote{Also called demographic parity or Independence of the predictions from the sensitive attribute \citep{barocas-hardt-narayanan}.} and equality of opportunity\footnote{The equality in true positive rates for both groups. This is a relaxation of equalized odds, also called Separation \citep{barocas-hardt-narayanan}.}. Throughout, we use logistic ridge regressions. We use the R-package \textsc{fairml} for the fairness-constrained risk scores \citep{Scutari2022} and do not require the fairness constraint to be met exactly.

All three methods, applying a decision threshold of $.5$, achieve an accuracy of about $64-65\%$. These results are in line with internationally reported accuracy rates for the prediction of long-term unemployment \citep{Desiere2019}. The unconstrained risk scores violate \textit{statistical parity}, with more women than men being predicted to become long-term unemployed (a discrepancy of $0.116$). Further, the true (a discrepancy of $0.174$) and false positive ($0.062$) rates are higher for women than for men. The fairness-constrained scores reduce these discrepancies. The unconstrained risk scores are approximately \textit{calibrated} for men and women, see Table~\ref{tab:risk_score_results}. Details on the implementation together with descriptive statistics for the risk scores can be found in Appendix~\ref{supp:risk}. 

\subsubsection{Prioritization}
For each of the three risk scores from the previous stage, we compile two priority lists modeling the Belgian and Austrian proposals. The Belgian list goes in order of decreasing risk \citep{Desiere2020}. The Austrian list does the same for those in the $30-70th$ risk percentiles. The others are put at the end of the list, in random order \citep{Allhutter2020}. This yields six priority lists, one for each combination of risk score and prioritization scheme. 

\subsubsection{Program Assignments}
For each of the six lists from the previous stage, we assign individuals to programs in order of priority. Individuals are assigned according to two schemes: optimal and random. The first assigns each person to the program that is most effective for them and not yet at capacity. This models the best-case scenario in which caseworkers are very good at discerning which program is best for each client. The second makes assignments by a uniform draw from the available programs.\footnote{We run this scheme ten times per policy and average over the resulting individual risks for long-term unemployment.} These two assignment schemes provide upper and lower bounds for what might happen when caseworkers are {\em informed} by risk scores when making assignment decisions instead of fully automating the decision. To model adjustments to the budget constraint of the PES, we consider the effect of increasing program capacities. As a baseline, we take the program sizes observed in the test set (see Table~\ref{tab:descriptives}). Then, we consider capacities that are $2-5x$ larger. Because the most effective programs are also the smallest, increasing overall capacities mainly influences outcomes by increasing the capacities of these small but effective programs.

\subsection{Results}
\subsubsection{Fair Prediction and the Fair Distribution of Social Goods.}
Regardless of the notion of retrospective fairness and the choices made at other stages, constraining risk predictions to be fair yields larger gender reemployment gaps (Figure~\ref{fig:fairnessPenalty}). This is because fairness constraints, by shifting the distribution of risk scores among women to look more like the distribution among men (Figure~\ref{fig:risk-scores}), tend to underestimate their risk of long-term unemployment. The effect of fairness constraints is to reserve a roughly equal number of seats in effective training programs for men and women (Figure~\ref{fig:number-participation}). Therefore, fairness-constrained policies induce similar improvements in labor market outcomes for both genders, which keeps the gender reemployment gap relatively constant. On the other hand, fairness unconstrained risk scores are, on average, higher for women. That means that more seats are reserved for women in effective programs--the result is more aggressive reductions in rates of long-term unemployment among women than among men. These effects are only made more pronounced when budget constraints are relaxed and program capacities are increased. For example, at baseline program sizes the combination of Belgian prioritization and individualized treatment decisions yields a $3.2\%$ gender gap in reemployment probabilities ($40.4\%$ vs $37.2\%$) when risk scores are unconstrained and a $4.1\%$ gender gap ($40.9\%$ vs $36.8\%$) when risk scores are constrained to satisfy equal opportunity. This means that, at baseline program sizes, the equal opportunity constraint slightly \textit{exacerbated} the ex-ante gender gap of $3.9\%$ ($43.6\%$ vs.  $39.7\%$). If programs are made five times larger, the fairness unconstrained policy reduces the gender gap to $.9\%$ ($35.1\%$ vs $34.2\%$) whereas equal opportunity leaves the gender gap relatively unchanged at $3\%$ ($36.2\%$ vs $33.2\%$). All results are given in Tables~\ref{tab:LTU-rates-1} for baseline and \ref{tab:LTU-rates-5} for five-fold capacities. We observe similar patterns for citizenship gaps (Appendix, Figures~\ref{fig:fairnessPenaltyCitizens} and \ref{fig:CitizenOverall}). 
\begin{figure}[hbt!]
  \centering

  \begin{subfigure}{0.48\textwidth} % 0.48\textwidth
    \includegraphics[width=0.9\linewidth]{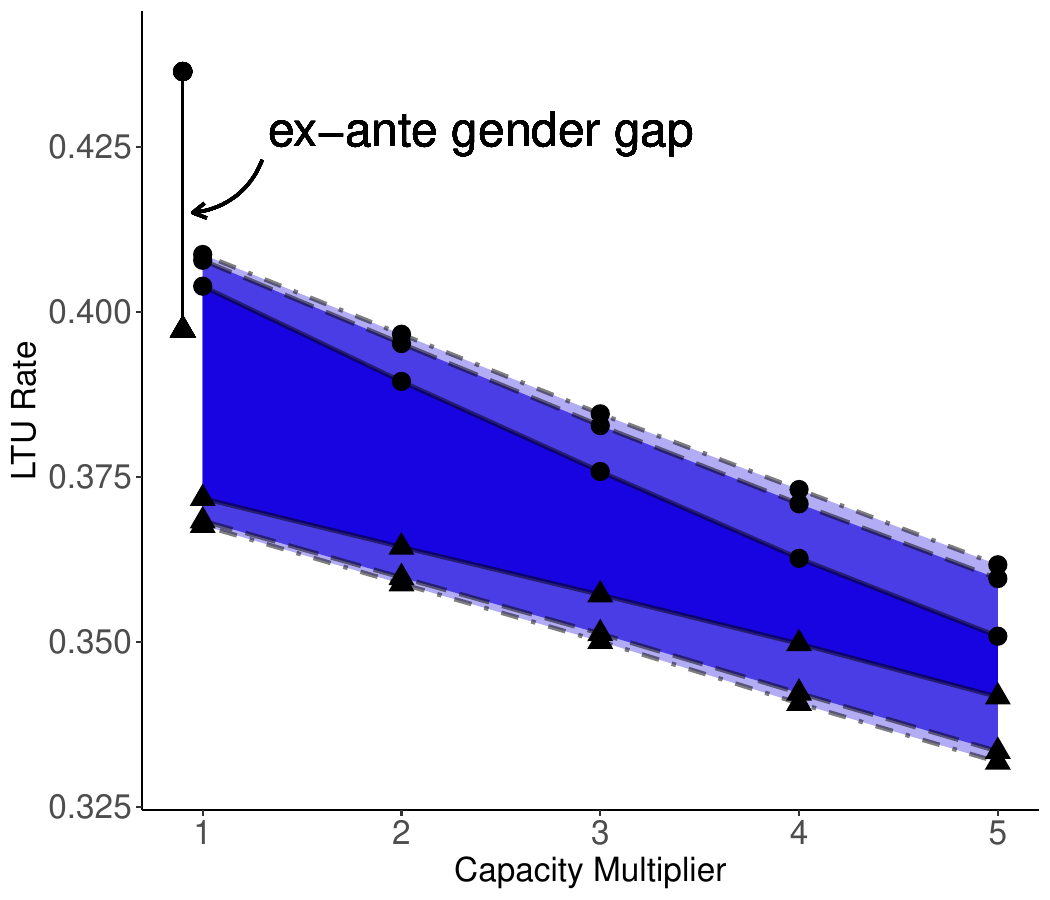}
    \caption{Belgian prioritization and optimal program.}
  \end{subfigure}
  \hfill
  \begin{subfigure}{0.48\textwidth}
    \includegraphics[width=0.9\linewidth]{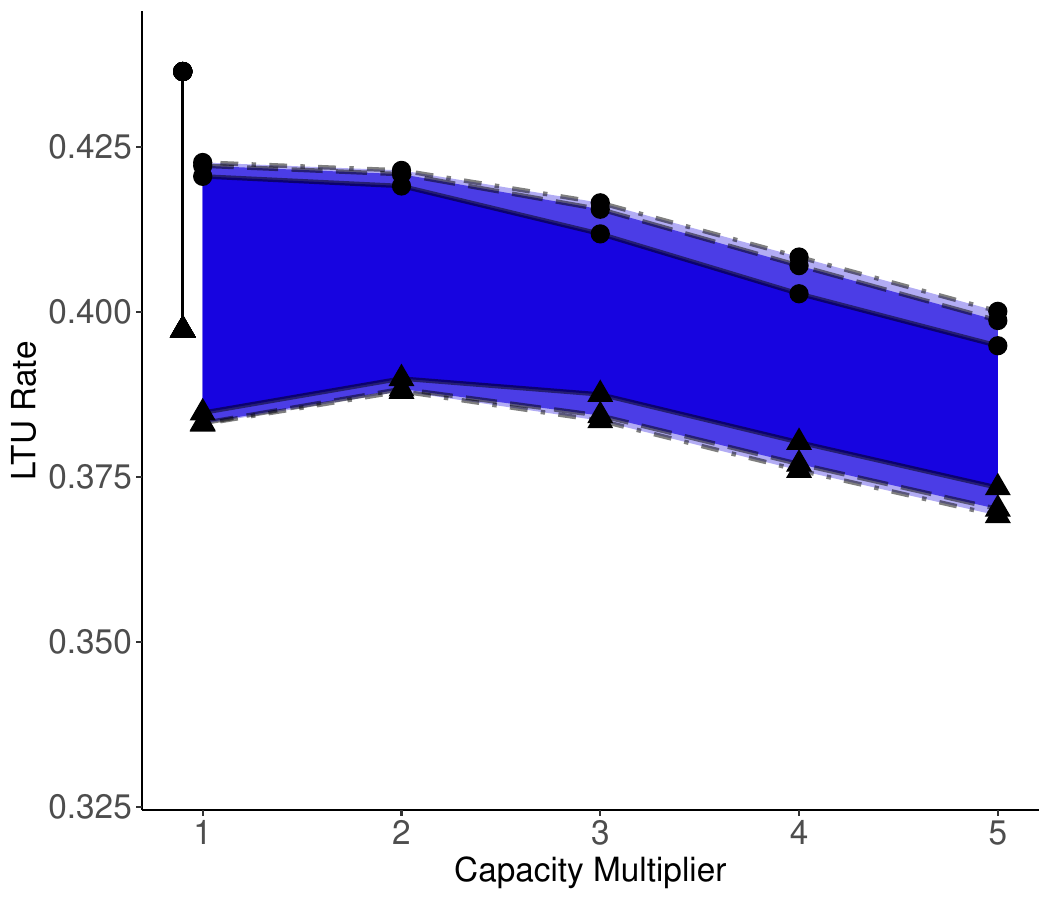}
    \caption{Belgian prioritization and random program.}
  \end{subfigure}
\medskip
  \begin{subfigure}{0.48\textwidth}
    \includegraphics[width=0.9\linewidth]{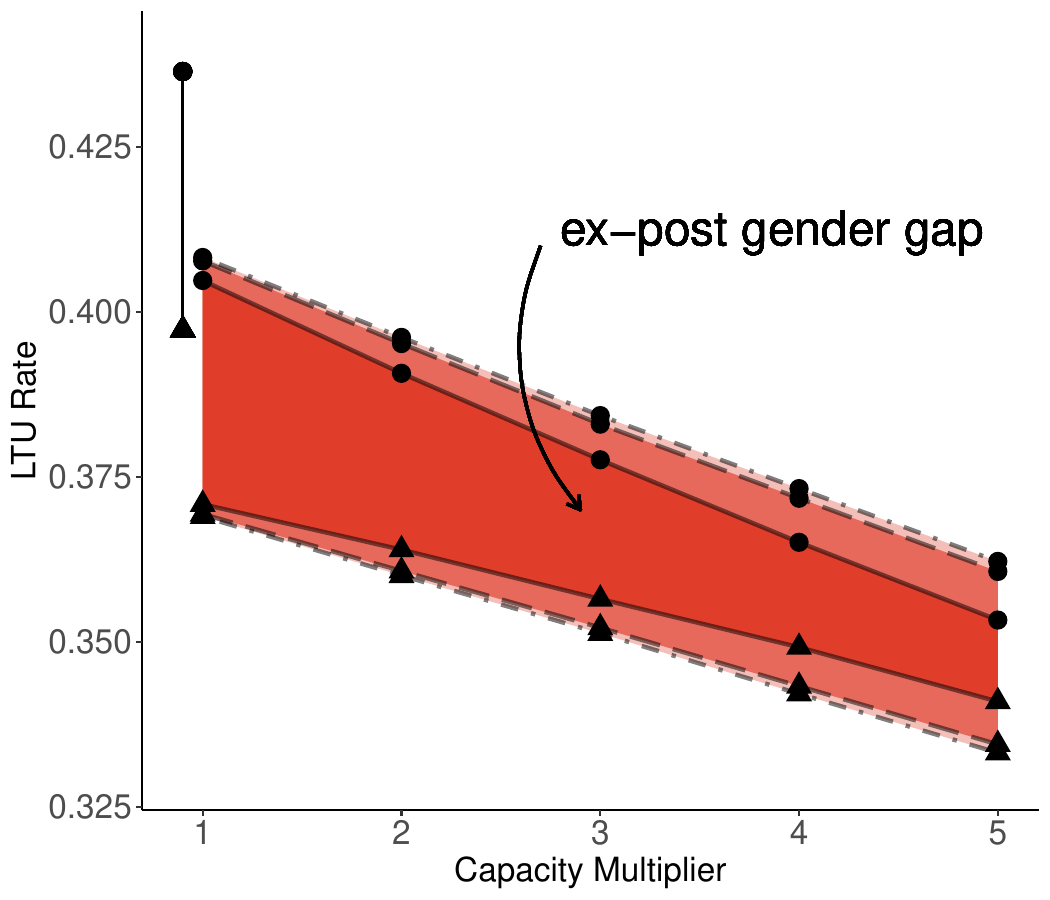}
    \caption{Austrian prioritization and optimal program.}
  \end{subfigure}
  \hfill
  \begin{subfigure}{0.48\textwidth}
    \includegraphics[width=0.9\linewidth]{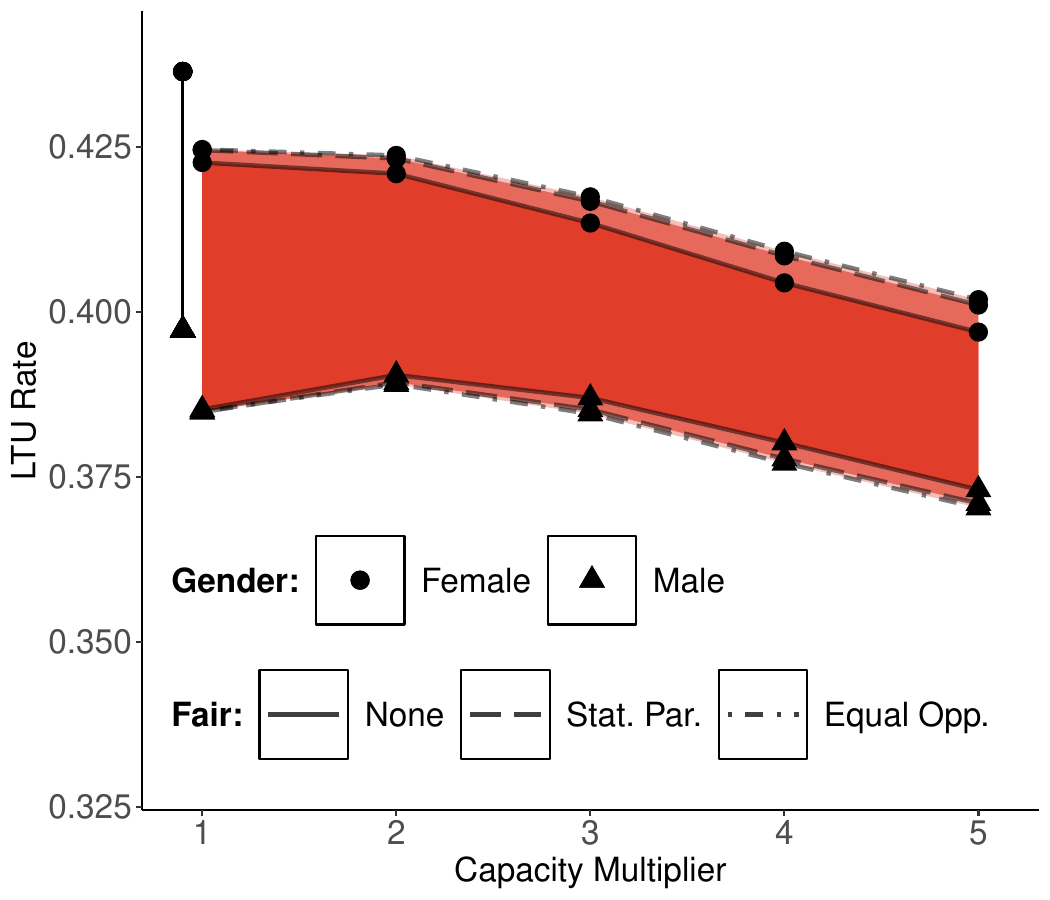}
    \caption{Austrian prioritization and random program.}
  \end{subfigure}
\caption{We plot the gender gap in long-term unemployment (LTU) against program capacity for each combination of prioritization and assignment scheme. The level of transparency shows the gender gap for the corresponding fairness constraint: none, statistical parity, or equal opportunity. The unconstrained risk scores (lowest transparency) result in the smallest gender gap. This effect is especially pronounced as program capacity is increased and program assignments are individualized (optimal).}  
\label{fig:fairnessPenalty}
\end{figure}

\subsubsection{Hawks and Doves.}
Regardless of other choices, the Belgian policy is at least as efficient as the Austrian policy, both in reducing overall rates of long-term unemployment and reducing the gender reemployment gap (Figure \ref{fig:doveSupremacy}). This holds both for the optimal program assignment and the random assignment. For example: at baseline program sizes, when the unemployed receive targeted assignment and risk scores are not fairness constrained, the Belgian policy achieves an overall LTU rate of $38.6\%$ and a gender reemployment gap of $3.2\%$ ($40.4\%$ vs. $37.2\%$) whereas the Austrian policy induces an identical overall rate and a gap of $3.4\%.$ If programs are made five times larger, the Belgian policy achieves an overall rate of $34.6\%$ and a gender gap of $.9\%$ ($35.1\%$ vs $34.2\%$), whereas the Austrian policy achieves an identical overall rate and a gender gap of $1.2\%$ ($35.3\%$ vs $34.1\%$). Thus, targeting those at the highest risk of long-term unemployment achieves improvements in gender equality without any costs in overall efficiency. A more fine-grained analysis shows that the Belgian prioritization closes the gender gap much more aggressively among married non-citizens, who tend to have the worst labor market outcomes, whereas the Austrian prioritization does slightly better among groups with better average outcomes (Figure~\ref{fig:married-swiss}). Similar effects are observed for citizenship gaps (Figures~\ref{fig:fairnessPenaltyCitizens} and \ref{fig:CitizenOverall}). Therefore we do not find any efficiency advantage for withholding training from individuals at the highest risk of unemployment. Indeed, risk scores tend to overestimate the risk of unemployment under optimal treatment (Figure~\ref{fig:risk-vs-po}).
\begin{figure}[hbt!]
  \centering

  \begin{subfigure}{0.48\textwidth}
    \includegraphics[width=.9\linewidth]{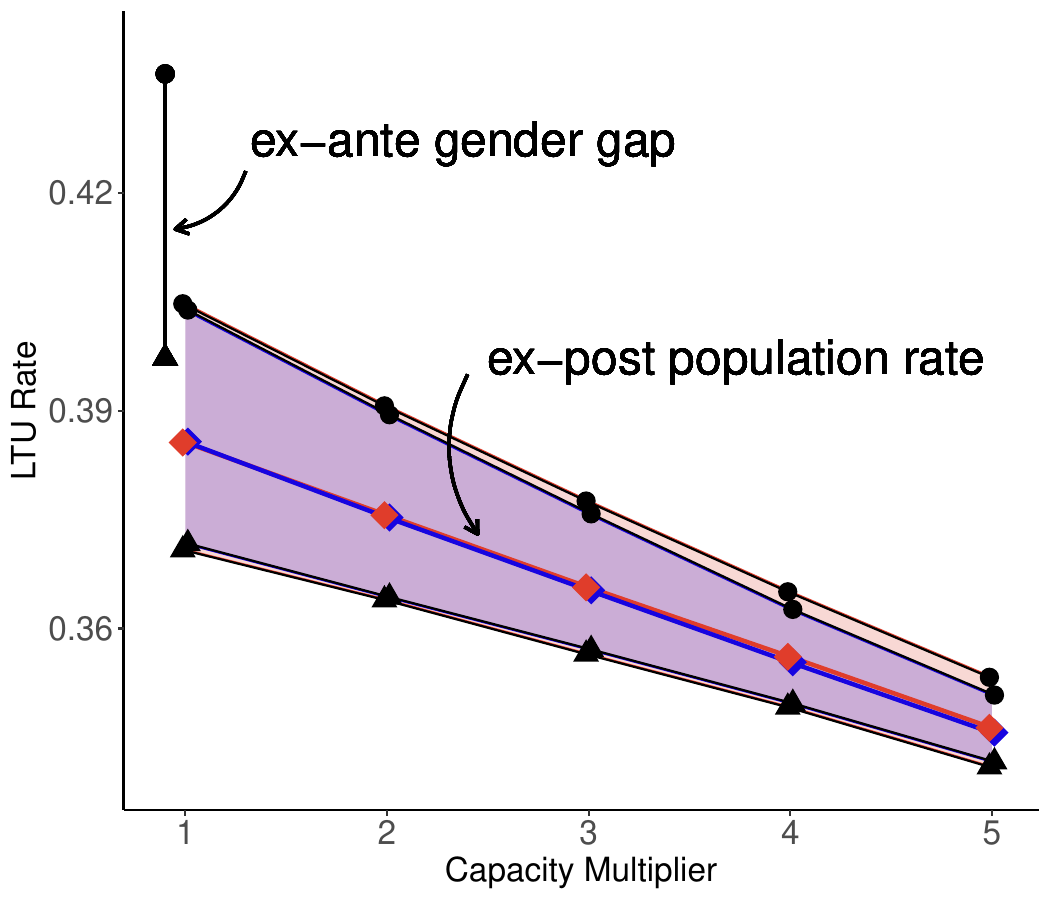}
    \caption{Optimal Program.}
    \label{fig:figure1}
  \end{subfigure}
  \hfill
  \begin{subfigure}{0.48\textwidth}
    \includegraphics[width=.9\linewidth]{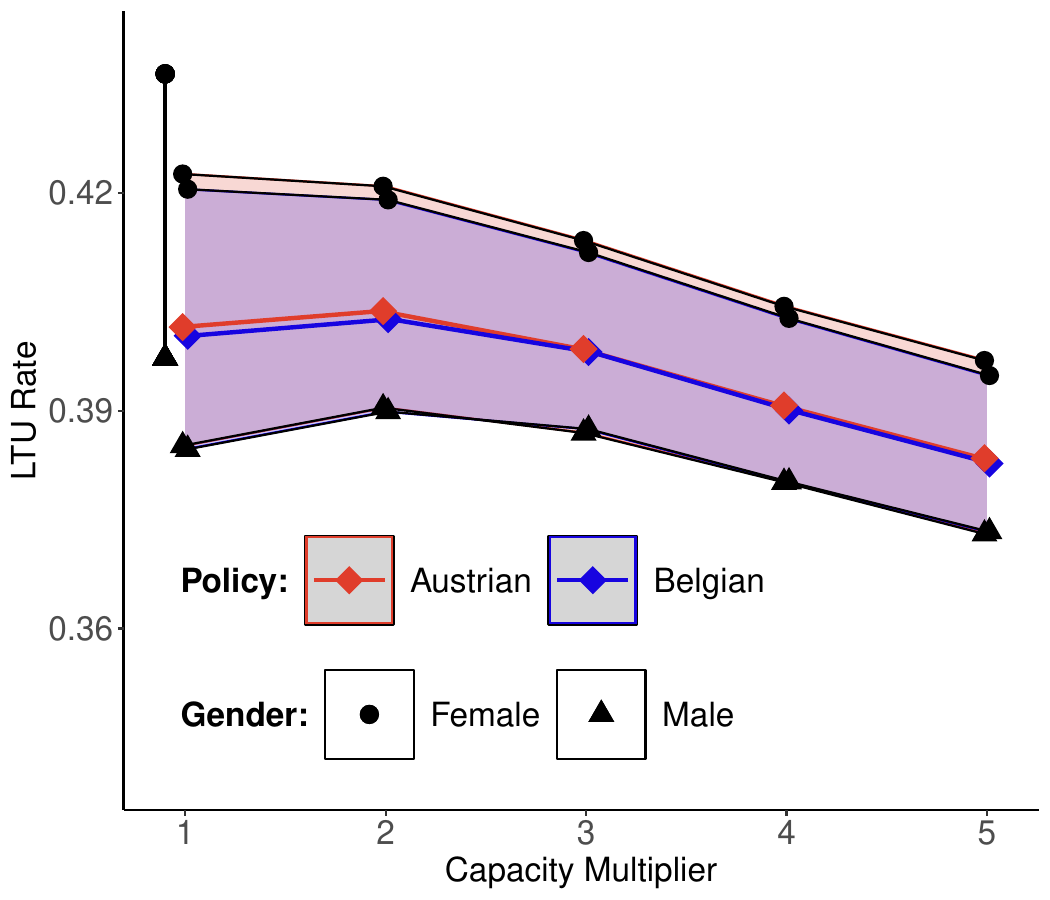}
    \caption{Random Program.}
    \label{fig:figure2}
  \end{subfigure}

  \caption{We show overall long-term unemployment and the gender gap against program capacity for each combination of prioritization and assignment scheme. For clarity, results are shown only for fairness-unconstrained risk scores. Regardless of the assignment scheme, the Belgian prioritization results in slightly lower overall rates of long-term unemployment (blue line) and a smaller gender gap. Individualized program assignments (optimal) are markedly more effective.}
  \label{fig:doveSupremacy}
\end{figure}

\subsubsection{Gains from Modeling Counterfactual Outcomes.}
Regardless of other choices, assigning individuals to the program with the highest estimated effectiveness reduces overall long-term unemployment and reemployment gaps (Figures \ref{fig:fairnessPenalty} and \ref{fig:doveSupremacy}). This represents gains due to explicit estimation of treatment effects rather than risk scores alone. For example: at baseline program sizes, when risk scores are not fairness constrained, targeting achieves a reduction of about $1.5$ percentage points in overall long-term unemployment over random assignment, regardless of prioritization. If programs are made five times larger, targeting achieves a reduction of about $3.7$ percentage points over random assignment. Targeting is also more effective than random assignment at reducing gender gaps under both prioritization regimes.

%%%%%%%%%%%%%%%%%%%%%%%%%%%%%%%%%%%%%%%%%%%%%%%%%%%%%%%%%%%%%%%%%%%%%%%%%%%%%%%%%%%%
\section{Conclusion, Limitations, and Future Work}\label{sec:Conclusion}
We have argued that {\em prospective} algorithmic fairness requires anticipating the causal effects of deploying algorithms on the distribution of outcomes. We have shown that existing methods in algorithmic fairness can have perverse distributive effects: requiring risk scores to be fair according to statistical parity or equal opportunity may exacerbate inequalities in social goods. Moreover, contrary to the accepted trade-offs between accurate and fair predictions, accurate prediction of individualized \textit{counterfactual} outcomes supports policy in reducing inequality in the distribution of social goods. 

Our approach has several limitations: we have not tried every fairness constraint (notably, multi-calibration \citep{HebertJohnson2018_CONF, Hoeltgen2023_CONF}), nor accounted for uncertainty in the estimation of individualized treatment effects and outcomes. Uncertainty quantification in double-robust machine learning remains an open problem \citep{Curth2024}. Conformal prediction methods may apply \citep{Alaa2023, Lei2021}. 
Some applications may require program assignments to be made in an online, rather than a batch, fashion \citep{Yin2023_CONF}. In addition to anticipatory evaluations, algorithmic policy should be designed to support {\em ex-post} evaluation, for example by (partial) randomization. Our approach is rather paternalistic: future work should accommodate the preferences of the unemployed themselves. Finally, we rely essentially on risk scores to facilitate prioritization. This reflects the state of algorithmic policy. However, risk scores increasingly seem like an unnecessary detour. We are inspired by \citet{Koertner2023}: future work might directly seek distributively optimal allocations (perhaps with more sophisticated notions of optimality) without recourse to risk scores \citep{Kitagawa2019, Viviano2023}. This approach subjects claims of `efficiency' to direct test and allows the conceptual innovations of distributive justice theory to flow directly into applications.  

\

\section*{Acknowledgments}
This work has been funded by the Deutsche Forschungsgemeinschaft (DFG, German Research Foundation) under Germany’s Excellence Strategy – EXC number 2064/1 – Project number 390727645. The authors thank the International Max Planck Research School for Intelligent Systems (IMPRS-IS) for supporting Sebastian Zezulka.

Many thanks to Michael Knaus, Christoph Kern, Ruben Bach, Thomas Grote, Donal Khosrowi Djen-Gheschlaghi, and the anonymous reviewers for helpful discussions and feedback, and to John Körtner and Ruben Bach for sharing their code.

%\bibliographystyle{ACM-Reference-Format}
\bibliography{refrences}

%\clearpage

\appendix

\section{Proof of Theorem~\ref{thm:Identification}}\label{supp:proof}
\begin{proof}[Proof of Theorem~\ref{thm:Identification}]
First, we need to show that all terms are well-defined. This amounts to showing that $P_\text{post}(A=a,X=x)$, $P_{\text{pre}}(A=a)$ and $P_{\text{pre}}(A=a,X=x,D=d)$ are strictly greater than zero for all $(x,d)\in \Pi_{\text{post}}.$ 

We first show that $P_{\text{pre}}(A=a)>0.$ Note that 

\begin{align*}
P_{\text{pre}}(A=a) &= \sum_{x\in\mathcal{X}} P_{\text{pre}}(A=a,X=x)\\
&= \sum_{x\in\mathcal{X}} P_{\text{post}}(A=a,X=x) \tag{\textsc{No Feedback}}\\
&= P_{\text{post}}(A=a)>0.
\end{align*}

We now show that $P_\text{post}(A=a,X=x)>0$ for all $(x,d)\in\Pi_{\text{post}}.$ Note that  

\begin{align*}
P_\text{post}(A=a,X=x) &= P_{\text{post}}(A=a)\sum_{e\in \mathcal{D}} P_\text{post}(X=x,D=e|A=a)\\
&\geq P_{\text{post}}(A=a) P_\text{post}(X=x,D=d|A=a)>0.
\end{align*}

Finally, we show that $P_{\text{pre}}(A=a,X=x,D=d)>0$ for all $(x,d)\in\Pi_{\text{post}}.$ Since $P_{\text{pre}}(A=a)>0$, it suffices to show that $P_{\text{pre}}(X=x,D=d|A=a)>0$ for all $(x,d)\in\Pi_{\text{post}}.$ Accordingly, suppose that $(x,d)\in \Pi_{\text{post}}.$ Then 

\begin{align*}
P_{\text{post}}(X=x,D=d|A=a) ~=~ P_{\text{post}}(D=d|X=x,A=a)P_{\text{post}}(X=x|A=a)>0,
\end{align*}

which entails that both $P_{\text{post}}(D=d|X=x,A=a)>0$ and $P_{\text{post}}(X=x|A=a)>0.$ By \textsc{No Unprecedented Decisions}, $P_{\text{pre}}(D=x|X=x,A=a)>0$ and by \textsc{No Feedback} $P_{\text{pre}}(X=x|A=a)>0.$ Therefore,

\begin{align*}
P_{\text{pre}}(X=x,D=d|A=a)~=~P_{\text{pre}}(D=x|X=x,A=a)P_{\text{pre}}(X=x|A=a)>0;
\end{align*}
and the question of well-definedness is settled. 

Next, note that: $P_{\text{post}}(Y=y~|~A=a) =$

\begin{align*}
    &= \sum_{(x,d)\in\Pi_{\text{post}}} P_{\text{post}}(Y=y~|~A=a, X=x, D=d)P_{\text{post}}(X=x, D=d~|~A=a) \tag{\text{Total Probability}}  \\
    &= \sum_{(x,d)\in\Pi_{\text{post}}} P_{\text{post}}(Y=y~|~A=a, X=x, D=d)P_{\text{post}}(X=x|A=a)P_{\text{post}}(D=d~|~A=a, X=x)\\
    &= \sum_{(x,d)\in\Pi_{\text{post}}} P_{\text{post}}(Y=y~|~A=a, X=x, D=d)P_{\text{pre}}(X=x|A=a)P_{\text{post}}(D=d~|~A=a, X=x). \tag{\textsc{No Feedback}}
\end{align*}

Note that, whenever defined,

\begin{align*}
    P_{t}(Y=y~&|~A=a, X=x, D=d) = \\
    &= P_{t}\left(\sum_{e\in\mathcal{D}}Y^e\mathbbm{1}[D=e]=1~|~A=a, X=x, D=d\right) \tag{\textsc{Consistency}}  \\
    &=P_{t}\left(Y^d=y~|~A=a, X=x, D=d\right)  \\
    &=P_{t}\left(Y^d=y~|~A=a, X=x\right).   \tag{\textsc{Unconfoundedness}}
\end{align*}

Therefore,

\begin{align*}
    P_{\text{post}}(Y=y~&|~A=a, X=x, D=d) = \\
    &=P_{\text{post}}\left(Y^d=y~|~A=a, X=x\right)  \\
    &=P_{\text{pre}}\left(Y^d=y~|~A=a, X=x\right)   \tag{\textsc{Stable CATE}} \\
    &= P_{\text{pre}}\left(Y=y~|~A=a, X=x, D=d\right);
\end{align*}

and, therefore, $P_{\text{post}}(Y=y~|~A=a)=$ 
\begin{align*}
    =\sum_{(x,d)\in\Pi_{\text{post}}} &P_{\text{pre}}(Y=y~|~A=a, X=x, D=d)P_{\text{pre}}(X=x|A=a)P_{\text{post}}(D=d~|~A=a, X=x).
\end{align*}
\end{proof}

\ifdim\columnwidth=\textwidth
\clearpage
\fi

\section{Case Study}
\subsection{Replication}
The replication package is available online on Github: \url{https://github.com/sezezulka/2023-01-ALMP-LTU.git}. It contains the code to run the pre-processing, the estimation of the individualized potential outcomes, the estimation of the (fairness constraint) risk scores, and the simulations of the algorithmically informed policies as described here. It allows the reproduction of the reported results, tables, and figures. Unfortunately, we are not allowed to make the data publicly available. It is available as a scientific use file on SWISSbase \citep{Lechner2020}.

%\clearpage
%\subsection{Descriptive Statistics: Simulation Data}\label{supp:descriptives}
\begin{table*}[hbt!]
    \centering
    \begin{tabular}{lccccccc}
    \toprule
    & \#Obs & LTU & Female & Age & Non-Citizen & Employability & Past Income \\
    & & &  (binary) & in years & (binary) & & in CHF \\
    \midrule
    \midrule
    \textbf{Simulation Data} & 32,148 & 0.41 & 0.44 & 36.8 & 0.36 & 1.93 & 43,461 \\
    \midrule
    No program & 23,785 & 0.41	& 0.43 & 36.6 & 0.37 & 1.92	& 42,557 \\
    Vocational & 423 & 0.28 & 0.32 & 37.5 & 0.32 & 1.91 & 49,349 \\
    Computer & 446 & 0.24 & 0.61 & 38.9 & 0.20 & 1.98 & 43,251 \\
    Language & 723 & 0.48 & 0.54 & 35.3 & 0.68 & 1.83 & 37,779 \\
    Job Search & 5,868 & 0.43 & 0.44 & 37.4 & 0.33 & 1.98 & 46,815 \\
    Employment & 321 & 0.46 & 0.43 & 35.3 & 0.39 & 1.84 & 36,902 \\
    Personality & 582 & 0.37 & 0.35 & 39.4 & 0.25 & 1.93 & 53,136 \\
    \bottomrule
    \end{tabular}
    \caption{Descriptive statistics for key demographic variables in the test and simulation data and by observed treatment groups. Long-term unemployment (LTU), Female, and Non-Citizen are given as shares. Age, Employability, and Past Income are averages. Employability is an ordered variable from low (1) to high (3), assigned by the caseworker. \citet{Knaus2022} reports an exchange rate USD/CHF of about $1.3$ for 2003.}
    \label{tab:descriptives}
\end{table*}

%\clearpage

%\subsection{Descriptive Statistics: Full sample}\label{supp:descriptives-full}
\begin{table*}[htb!]
    \centering
    \begin{tabular}{lccccccc}
    \toprule
    & \#Obs & LTU & Female & Age & Non-Citizen & Employability & Past Income \\
    & & &  (binary) & in years & (binary) & & in CHF \\
    \midrule
    \midrule
    \textbf{Full Sample} & 64,296 & 0.41 & 0.44 & 36.8 & 0.36 & 1.93 & 43,391 \\
    \midrule
    No program & 47,631 & 0.41 & 0.44 & 36.6 & 0.37 & 1.93 & 42,529 \\
    Vocational & 858 & 0.29 & 0.33 & 37.5 & 0.30 & 1.93 & 48,654 \\
    Computer & 905 & 0.28 & 0.60 & 39.1 & 0.21 & 1.97 & 43,213 \\
    Language & 1,504 & 0.47 & 0.55 & 35.28 & 0.66 & 1.85 & 37,300 \\
    Job Search & 11,610 & 0.43 & 0.44 & 37.3 & 0.33 & 1.98 & 46,693 \\
    Employment & 611 & 0.43 & 0.41 & 35.3 & 0.38 & 1.83 & 37,084 \\
    Personality & 1,177 & 0.37 & 0.36 & 38.7 & 0.27 & 1.93 & 53,067 \\
    \bottomrule
    \end{tabular}
    \caption{Descriptive statistics for key demographic variables in the full sample and by observed treatment groups. The simulation data is drawn from this full sample. Long-term unemployment (LTU), Female, and Non-Citizen are given as shares. Age, Employability, and Past Income are averages. Employability is an ordered variable from low (1) to high (3), assigned by the caseworker. \citet{Knaus2022} reports an exchange rate USD/CHF of about $1.3$ for 2003.}
    \label{tab:descriptives-full}
\end{table*}

\subsection{Double-Robust Machine Learning for Estimating IAPOs}\label{supp:iapos}
In Section~\ref{sec:Theorem}, we have derived the formal conditions under which the post-interventional gender gap is identified. Two assumptions concern the internal validity of our study. \textsc{Unconfoundedness} is the strongest assumption. Replicating the work by \citet{Knaus2022, Knaus2022a} and \citet{Koertner2023}, we rely on extensive information on caseworkers and their subjective assessment of their clients in the estimation of treatment effects combined with rich administrative data on the demographics and employment biographies to support the assumption. \textsc{No Unprecedented decisions} requires that the propensity scores are non-zero. The other two concern the external validity of our simulation study. We presuppose that the treatment effects of the programs are stable under different allocations and increased program capacities (\textsc{stable CATEs}) and that the pool of unemployed stays the same (\textsc{No Feedback} on the covariates). 

First, we estimate the normalized conditional probability to be allocated into each program (the propensity of treatment, $e_d(X_i)$) and the conditional outcome mean in the observed allocation (in short, conditional outcome, $\mu(d, x)$). Given the small number of observations in most of the labor market programs, we use the full data set and cross-validation for the estimation of the nuisance parameters. The two nuisance parameters then allow the estimation of the doubly robust score:
$$ \hat{\Gamma}_{i,d} = \hat{\mu}(d,X_i) + \frac{D_i(d)(Y_i - \hat{\mu}(d,X_i))}{\hat{e}_d(X_i)},$$
where $D_i(d)$ indicates the treatment assignment for individual $i$ and $Y_i$ the observed, pre-inter- ventional outcome. This strategy is called doubly robust because the functional form of either the propensity score or the conditional outcome can be miss-specified without threatening the identification \citep{Chernozhukov2018, Knaus2022}. In the last step, the estimates of the debiased scores, $\hat{\Gamma}_{i,d}$, are used as pseudo outcomes to estimate the conditional expected outcomes, $\mathbf{E}[\hat{\Gamma}_{i,d}~|~X_i]$ using a regression forest. These estimates are the individualized average potential outcomes for each treatment option under the outlined identifying assumptions. For this step, the regression forest is trained only on the training set.

We estimate individualized average treatment effects for each individual $i$ in the sample as differences between the respective individualized average potential outcomes: $$\hat{\Delta}_{i,d,d'} = \hat{\Gamma}_{i,d} - \hat{\Gamma}_{i,d'}.$$
In Table~\ref{fig:IATEs-by-Gender}, we show the distribution of individualized average treatment effects by gender. While the overall trends remain the same, all treatments except job search programs on average are slightly more effective for women than for men. Treatment effects are estimated against the baseline of no program.

\begin{figure*}[h]
    \centering
    \includegraphics[width=.8\linewidth]{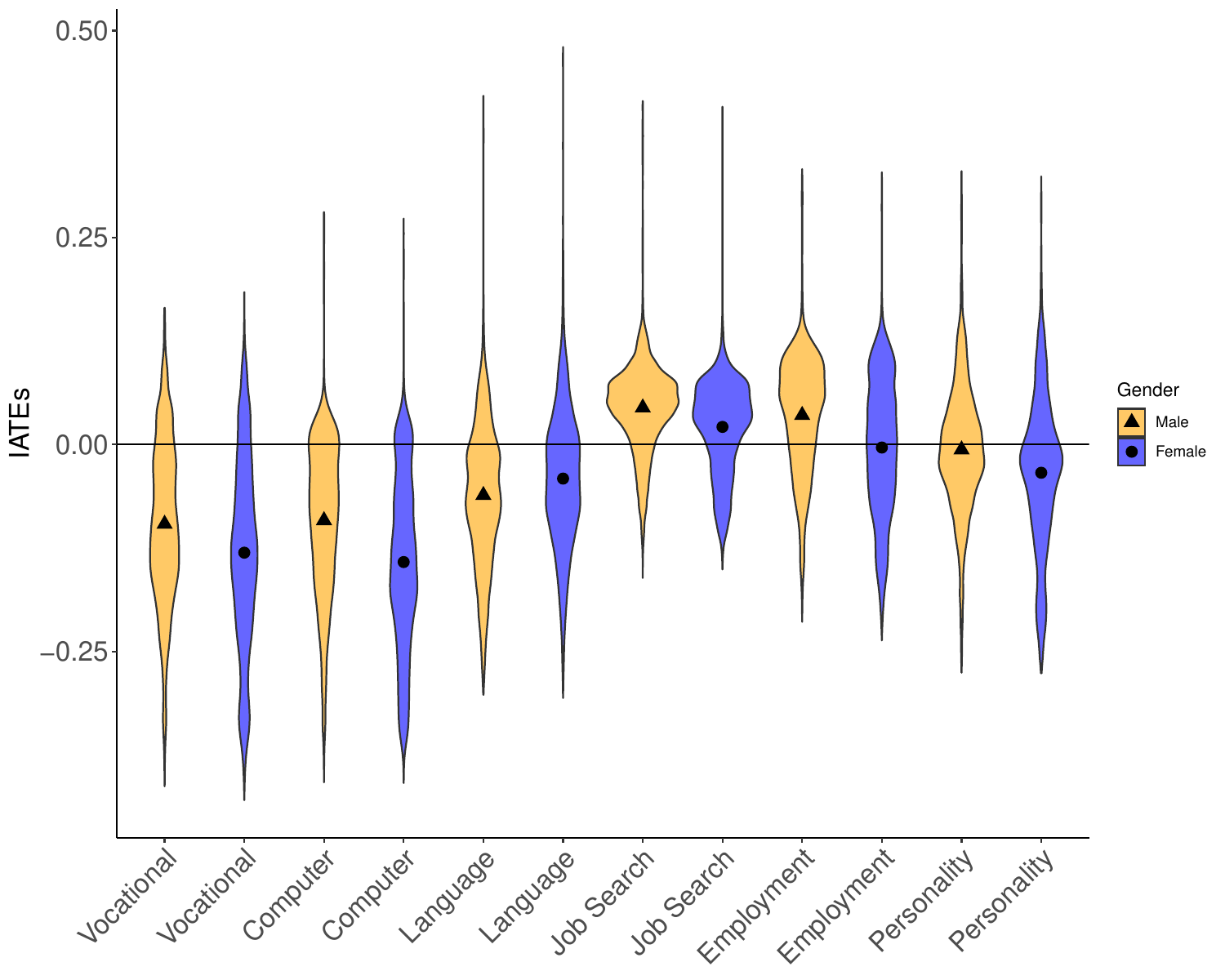}
    \caption{Individualized and Average Treatment Effects for six labor market programs by gender. The baseline treatment against which the treatment effects are estimated is ``no program''.}
    \label{fig:IATEs-by-Gender}
\end{figure*}

%\clearpage
\subsection{Risk Scores and Prioritization Policies}\label{supp:risk}
To determine the prioritization of registered unemployed in its Belgian or Austrian variants we estimate risk scores for becoming long-term unemployed. The full list of features is given in Table~\ref{tab:risk-features}. For a discussion on the predictability of long-term unemployment, see \citet{Mueller2023_TECH_REPORT}. Using administrative data from Germany, \citet{Kunaschk2022} evaluate the performance of risk scores under external shocks like the COVID-19 pandemic. \citet{Kern2021} evaluate the violation of retrospective fairness criteria when predicting long-term unemployment in the same context.

First, we estimate risk scores by a fairness-unconstrained logistic ridge regression. The optimal regularization strength is chosen by cross-validation at about $\lambda=0.049$. Second, we add a fairness constraint for {\em statistical parity} and, third, a constraint for {\em equal opportunity}. In this case, the true positive rates among the sensitive attribute are equalized, a relaxation of Separation \citep{Hardt2016}. We make use of the the implementation by \citep{Scutari2022} for the estimation of fairness-constrained risk scores. To achieve statistical parity they use a ridge penalty to bound the variance explained by the sensitive attribute (gender) over the total explained variance. For equal opportunity, the risk score is regressed against the sensitive attribute and the outcome variable with the ridge penalty bounding the variance explained by the sensitive attribute over the total explained variance. In both cases, we use a fairness penalty of $0.01$, where $0$ requires perfect fairness and $1$ corresponds to no fairness constraint. 

Note some important differences between the Belgian and Austrian implementations of our work. In Flanders, Belgium the probability of re-employment within six months is estimated by a random forest model \citep{Desiere2020}. Sensitive attributes are no longer included due to privacy regulations. In our simulation study, the definition of long-term unemployment corresponds to the ILO definition with $12$ months of uninterrupted unemployment. 

In Austria, two different models are estimated \citep{Allhutter2020_TECH_REPORT}. The first, short-term model, uses as a binary target at least $90$ days of unsupported employment within seven months after the reference date. The second, long-term model, uses at least $180$ days of unsupported employment within $24$ months as the target. Those with a short-term probability of employment above $66\%$ are classified as low risk for LTU. Those with a long-term probability of employment below $25\%$ are classified as high risk. The middle group is built as a residual. That is, it includes all those not classified as high or low risk. In difference to earlier reports \citep{Allhutter2020}, a stratification approach is applied, and logistic regressions are used to evaluate the feature importance only \citep{Allhutter2020_TECH_REPORT}. Sensitive attributes like gender and citizenship are included as features. In difference to the Austrian proposal, we estimate one model and create the prioritized middle group as those individuals falling in the $30-70th$ percentile of the respective risk distribution.  

\subsection{Further Results}
Following, we present several additional results. In Table~\ref{tab:descriptives}, we report descriptive statistics for our simulation and test data with $32,148$ observations. Table~\ref{tab:descriptives-full} shows the respective statistics for the full dataset. 

Figure~\ref{fig:IATEs-by-Gender} compares the distribution of the estimated Individualized Average Treatment Effects (IATEs) by gender. 

Figure~\ref{fig:risk-scores} shows histograms of all three risk scores, fairness constraint and not. Results on the predictive power of the risk scores after applying a decision threshold at $0.5$ and a formal fairness analysis with gender as the sensitive attribute are reported in Table~\ref{tab:risk_score_results}.

Rates of long-term unemployment under the different algorithmically informed policies and baseline as well as five-fold capacities are reported in Tables~\ref{tab:LTU-rates-1} and \ref{tab:LTU-rates-5}. Program participation by gender under the algorithmically informed policies is shown in Figure~\ref{fig:number-participation}.

We show the gap in long-term unemployment between Swiss citizens and non-citizens for each combination of prioritization and assignment schemes in Figure~\ref{fig:fairnessPenaltyCitizens}. As for gender, the comparison with the overall LTU rate is presented in Figure~\ref{fig:CitizenOverall}.

Figure~\ref{fig:married-swiss} shows both the gender gaps in long-term unemployment and overall LTU rates for four subgroups in our data: unmarried non-citizen, married non-citizen, unmarried Swiss citizen, and married Swiss citizen. 

Lastly, in Figure~\ref{fig:risk-vs-po} we plot the estimated risk scores against the respective optimal, that is lowest, potential outcome. 

\begin{table*}[tb]
    % \centering
    \begin{tabular}{lc}
     \toprule
    Features used for the estimation of risk scores & \\
     \midrule \midrule
      Age   &  \\
      Mother tongue in canton's language  &  \\
      Lives in big city &  \\
      Lives in medium city &  \\
      Lives in no city &  \\
      Fraction of months employed in last 2 years &  \\
      Number of employment spells in last 5 years &  \\
      Female (binary) &  \\
      Foreigner with temporary permit &  \\
      Foreigner with permanent permit &  \\
      Cantonal GDP p.c. &  \\
      Married &  \\
      Mother tongue other than German, French, Italian &  \\
      Past income in CHF &  \\
      Previous job: Manager &  \\
      Previous job in missing sector &  \\
      Previous job in primary sector &  \\
      Previous job in secondary sector &  \\
      Previous job in tertiary sector &  \\
      Previous job: self-employed &  \\
      Previous job: skilled worker &  \\
      Previous job: unskilled worker &  \\
      Qualification: semiskilled &  \\
      Qualification: some degree &  \\
      Qualification: unskilled &  \\
      Qualification: skilled without degree &  \\ 
      Swiss citizenship &  \\
      Number of unemployment spells in last 2 years &  \\
      Cantonal unemployment rate in \% &  \\
    \bottomrule
    \end{tabular}
    \caption{List of features from the ``Swiss Active Labor Market Policy Evaluation Dataset'' \citep{Lechner2020} used for the estimation of risk scores. All caseworker information is omitted, sensitive attributes like ``Female'' or ``Citizenship'' are included.}
    \label{tab:risk-features}
\end{table*}

\begin{table*}[tb]
    \centering
    \begin{tabular}{rcccc}
    \toprule
        & Reference & Ridge Regression & Stat. Parity & Equal Opportunity \\
    \midrule
    \midrule
    Accuracy & (1) & 0.644 & 0.644 &  0.645 \\
    Precision & (1) & 0.612 &  0.605 & 0.607 \\
    Recall & (1) & 0.384 & 0.404 & 0.404 \\ \midrule
    Stat. Parity & (0) & \textbf{0.116} & \textbf{0.041} & 0.019 \\
    Equal Opportunity & (0) & \textbf{0.173} & 0.07 & \textbf{0.044} \\
    False Positive Parity & (0) & \textbf{0.062} & 0.005 & -0.014 \\
    Pos. Predictive Parity & (0) & 0.062 & 0.072 & 0.081 \\
    Neg. Predictive Parity & (0) & 0.011 & 0.011 & 0.016 \\
    \bottomrule
    \end{tabular}
    \caption{Results for predicting long-term unemployment (LTU). To get binary predictions of the target, a threshold of $0.5$ is applied to the risk scores. The sensitive attribute for the group-based fairness analysis is gender. All results are reported as differences between the respective results for \textit{women}  and \textit{men}. For example, statistical parity is estimated as $P(\hat{Y}=1~|~A=1)-P(\hat{Y}=1~|~A=0)$, where $\hat{Y}$ is the random variable representing the binary predictions of LTU and $A$ the sensitive attribute.}
    \label{tab:risk_score_results}
\end{table*}

\begin{figure*}[h]
    \centering

    \begin{subfigure}[b]{0.45\textwidth}
        \centering
        \includegraphics[width=\linewidth]{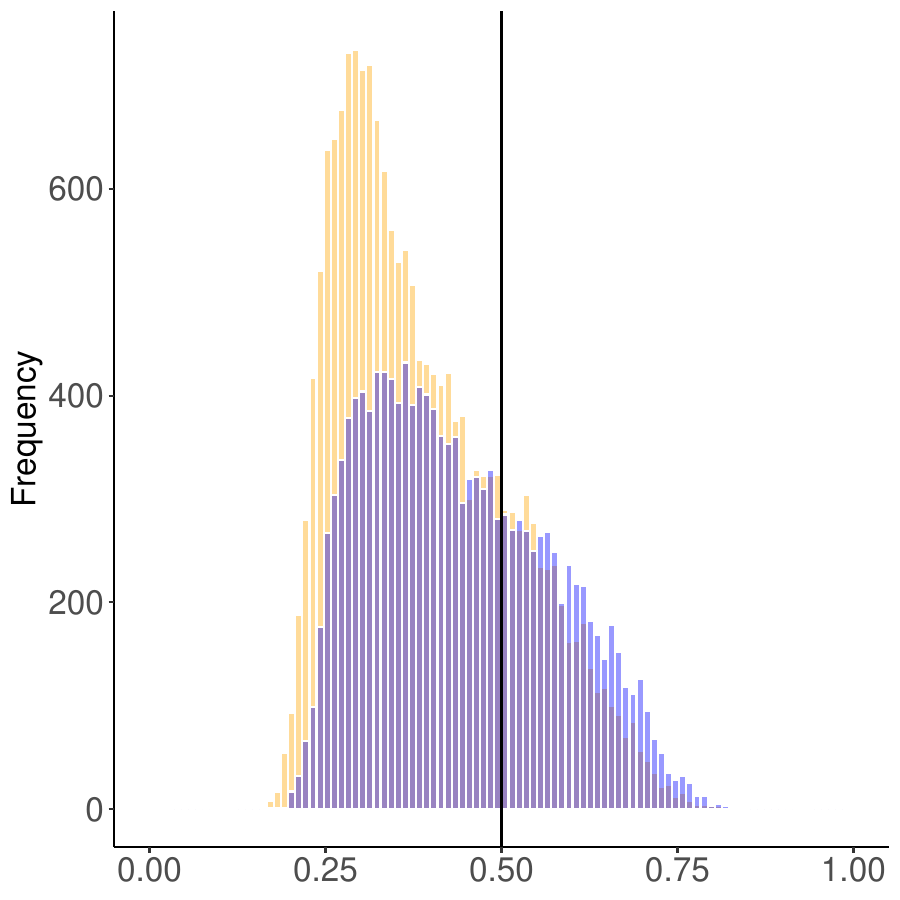}
        \caption{No fairness constraint (approximately calibrated by Gender).}
        \label{fig:risk-unconstraint}
    \end{subfigure}
    \hfill
    \begin{subfigure}[b]{0.45\textwidth}
        \centering
        \includegraphics[width=\linewidth]{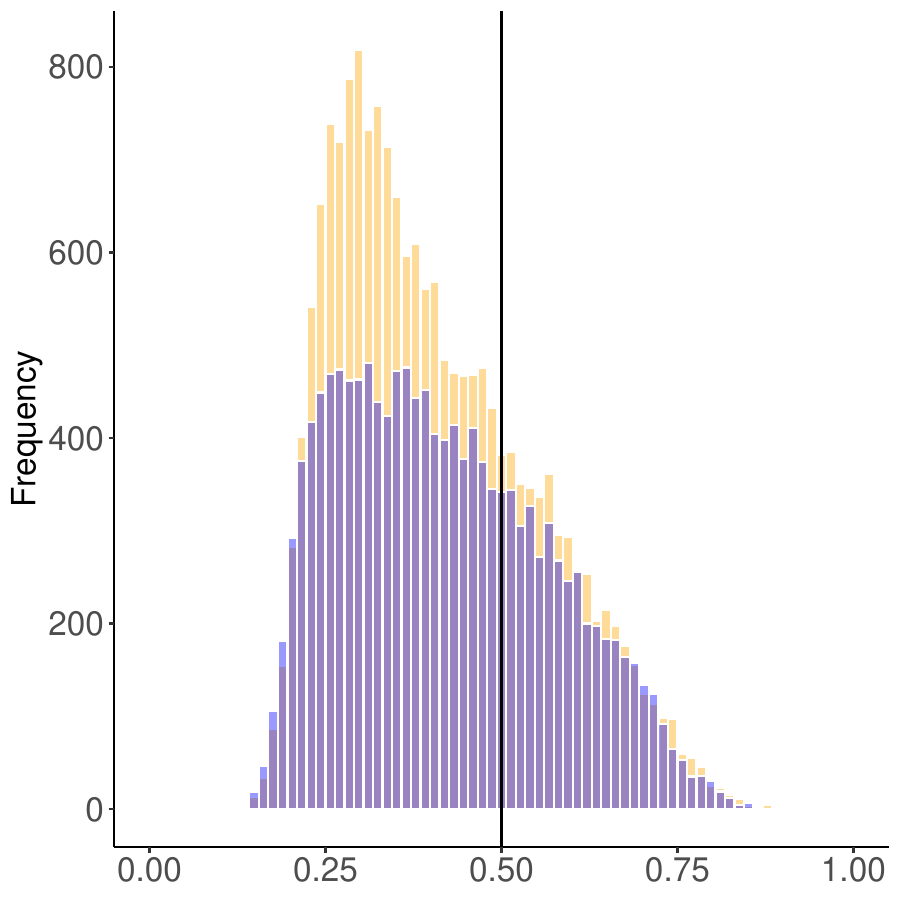}
        \caption{Statistical parity constraint.}
        \label{fig:risk-statparity}
    \end{subfigure}
    \hfil
    \begin{subfigure}[b]{0.45\textwidth}
        \centering
        \includegraphics[width=\linewidth]{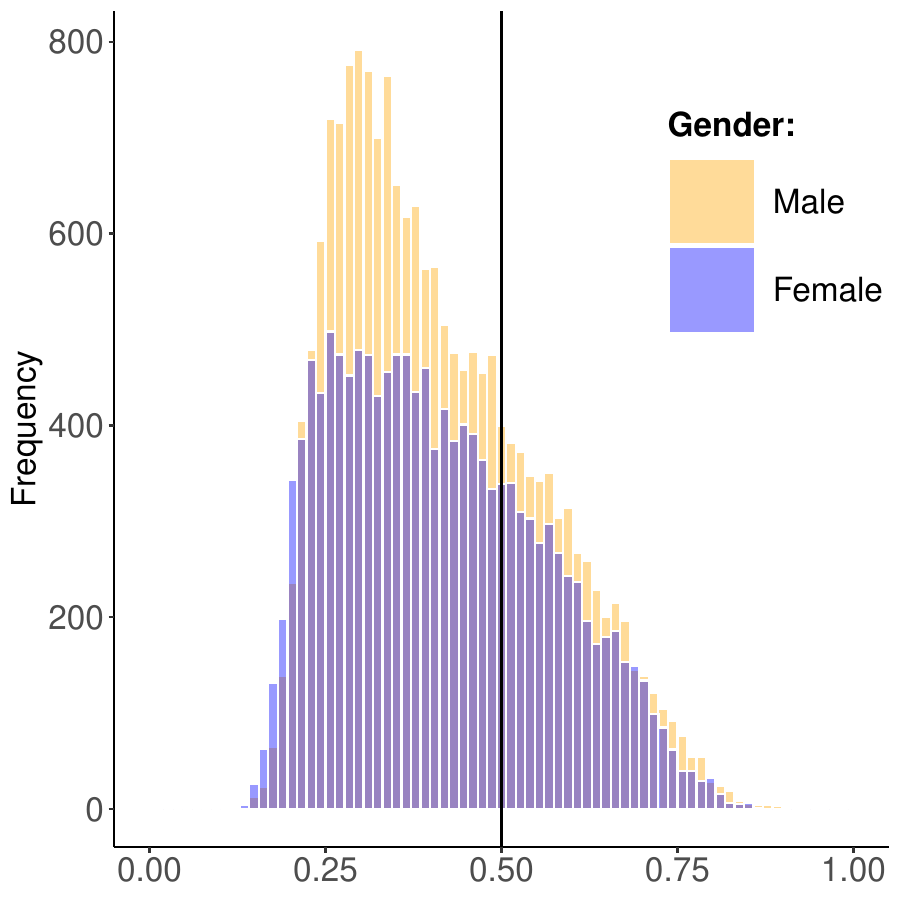}
        \caption{Equal opportunity constraint.}
        \label{fig:risk-equalopp}
    \end{subfigure}

    \caption{Risk scores for long-term unemployment by gender, estimated by logistic ridge regression with and without fairness constraints. The vertical line at $.5$ gives the decision threshold for binary predictions.}
    \label{fig:risk-scores}
\end{figure*}

\clearpage
%\subsection{Results from the Simulation Study}\label{supp:results}

\begin{table}[tb!]
    \centering
    \begin{tabular}{lccccccc}
    \toprule
    & LTU & Women & Men & Gender gap & Non-Citizens & Citizen & Citizen Gap \\
    \midrule \midrule
    Status quo  &  0.414 & 0.436 & 0.397 & 0.039 & 0.515 & 0.357 & 0.158  \\
    \midrule
    \textbf{Belgian, optimal} & \\
    Logistic Regression & 0.386 & 0.404 & 0.372 & 0.032 & 0.446 & 0.351 & 0.095 \\
    Stat. Parity & 0.386 & 0.408 & 0.368 & 0.039 & 0.448 & 0.35 & 0.097 \\
    Equal Opp. & 0.386 & 0.409 & 0.368 & \textbf{0.041} & 0.448 & 0.35 & 0.097 \\
    \midrule
    \textbf{Belgian, random} & \\
    Logistic Regression & 0.4 & 0.421 & 0.385 & 0.036 & 0.473 & 0.359 & 0.114 \\
    Stat. Parity & 0.400 & 0.422 & 0.383 & 0.039 & 0.473 & 0.359 & 0.114 \\
    Equal Opp. & 0.400 & 0.423 & 0.383 & \textbf{0.04} & 0.474 & 0.359 & 0.115  \\
    \midrule
    \textbf{Austrian, optimal} & \\
    Logistic Regression & 0.386 & 0.405 & 0.371 & 0.034 & 0.447 & 0.351 & 0.097 \\
    Stat. Parity & 0.386 & 0.408 &  0.369 & 0.038 & 0.451 & 0.349 & 0.101 \\
    Equal Opp. & 0.386 & 0.408 & 0.369 & 0.039 & 0.451 & 0.349 & 0.101 \\
    \midrule
    \textbf{Austrian, random} & \\
    Logistic Regression & 0.402 & 0.423 & 0.385 & 0.037 & 0.476 & 0.359 & 0.117 \\
    Stat. Parity & 0.402 & 0.424 &  0.385 & \textbf{0.04} & 0.479 & 0.358 & 0.120 \\
    Equal Opp. &  0.402 & 0.425 & 0.385 & \textbf{0.04} & 0.479 & 0.359 & 0.120 \\
    \bottomrule
    \end{tabular}
    \caption{Rates of long-term unemployment (LTU) under the different algorithmically informed policies and \textbf{baseline} capacities.}
    \label{tab:LTU-rates-1}
\end{table}

\begin{table}[tb!]
    \centering
    \begin{tabular}{lccccccc}
    \toprule
    & LTU & Women & Men & Gender Gap & Citizens & Non-Citizen & Citizen Gap \\
    \midrule \midrule
    Status quo  & 0.414 & 0.436 & 0.397 & 0.039 & 0.515 & 0.357 &  0.158 \\
    \midrule
    \textbf{Belgian, optimal} & \\
    Logistic Regression & 0.346 & 0.351 & 0.342 & 0.009 & 0.375 & 0.329 & 0.046   \\
    Stat. Parity & 0.345 & 0.36 & 0.333 & 0.026 & 0.378 & 0.326 & 0.051 \\
    Equal Opp. & 0.345 & 0.362 & 0.332 & 0.03 & 0.377 & 0.326 & 0.051 \\
    \midrule
    \textbf{Belgian, random} & \\
    Logistic Regression & 0.383 & 0.395 &  0.373 & 0.022 & 0.44 & 0.350 & 0.09  \\
    Stat. Parity &  0.383 & 0.399 & 0.370 & 0.029 & 0.441 & 0.349 &  0.092  \\ 
    Equal Opp. & 0.383 & 0.400 & 0.369 & 0.031 & 0.442 & 0.349 & 0.092 \\
    \midrule
    \textbf{Austrian, optimal} & \\
    Logistic Regression & 0.346 & 0.353 & 0.341 & 0.012 & 0.380 & 0.327 & 0.053 \\
    Stat. Parity & 0.346 & 0.361 & 0.334 & 0.026 & 0.388 & 0.322 & 0.066\\
    Equal Opp. &  0.346 & 0.362 & 0.333 & 0.029 & 0.387 & 0.322 & 0.065 \\
    \midrule
    \textbf{Austrian, random} & \\
    Logistic Regression & 0.383 & 0.397 & 0.373 & 0.024 & 0.444 & 0.349 & 0.095  \\
    Stat. Parity &  0.384 & 0.401 & 0.371 & 0.030 & 0.449 & 0.347 & 0.102  \\
    Equal Opp. & 0.384 & 0.402 & 0.370 & 0.032 & 0.449 & 0.347 & 0.102 \\
    \bottomrule
    \end{tabular}
    \caption{Rates of long-term unemployment (LTU) under the different algorithmically informed policies and \textbf{five-fold} capacities.}
    \label{tab:LTU-rates-5}
\end{table}

\begin{figure}[htb!]
    \centering

    \begin{subfigure}{\textwidth}
        \includegraphics[width=\linewidth]{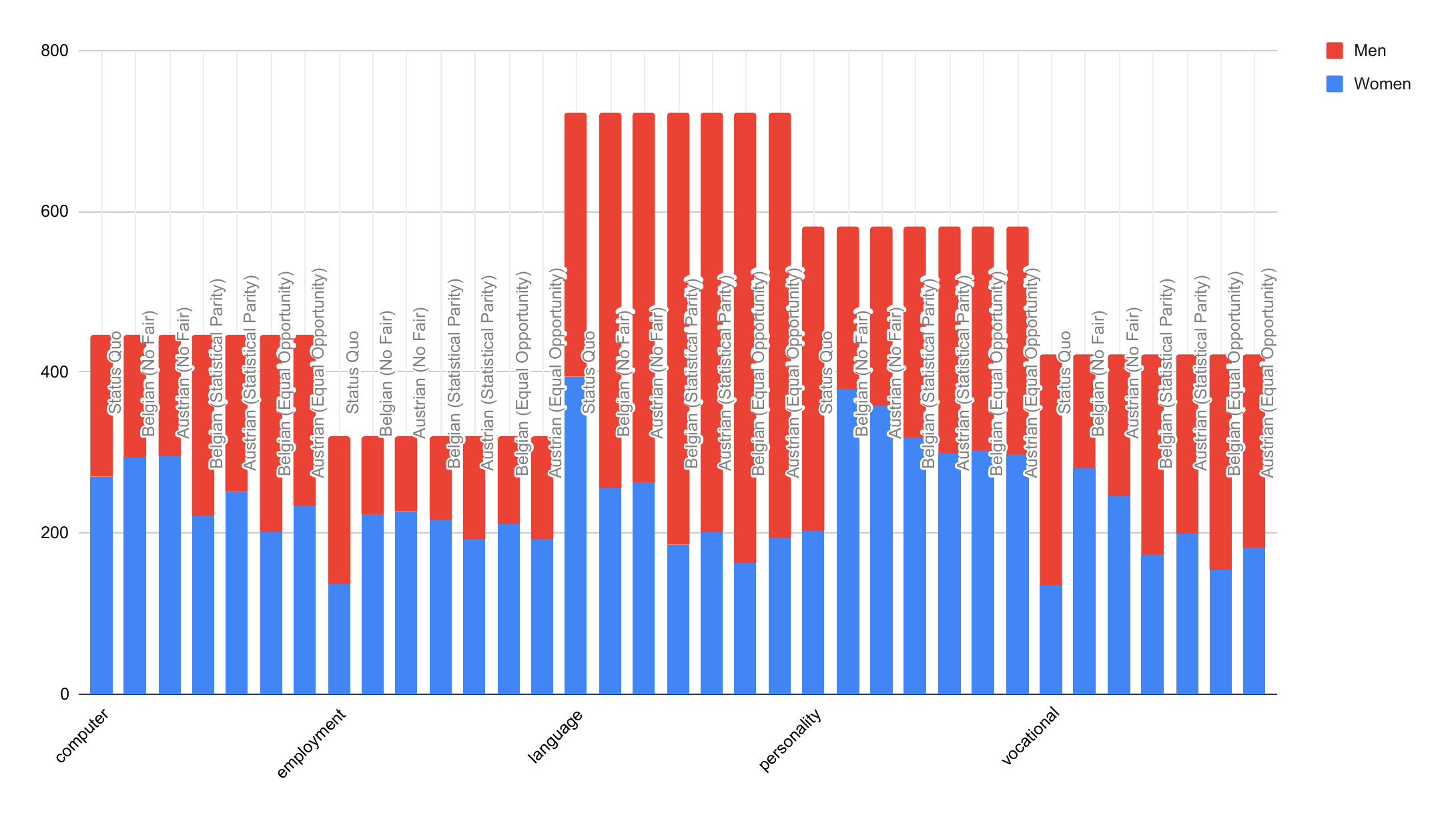}
    \end{subfigure}

    \begin{subfigure}{\textwidth}
        \includegraphics[width=.9\linewidth]{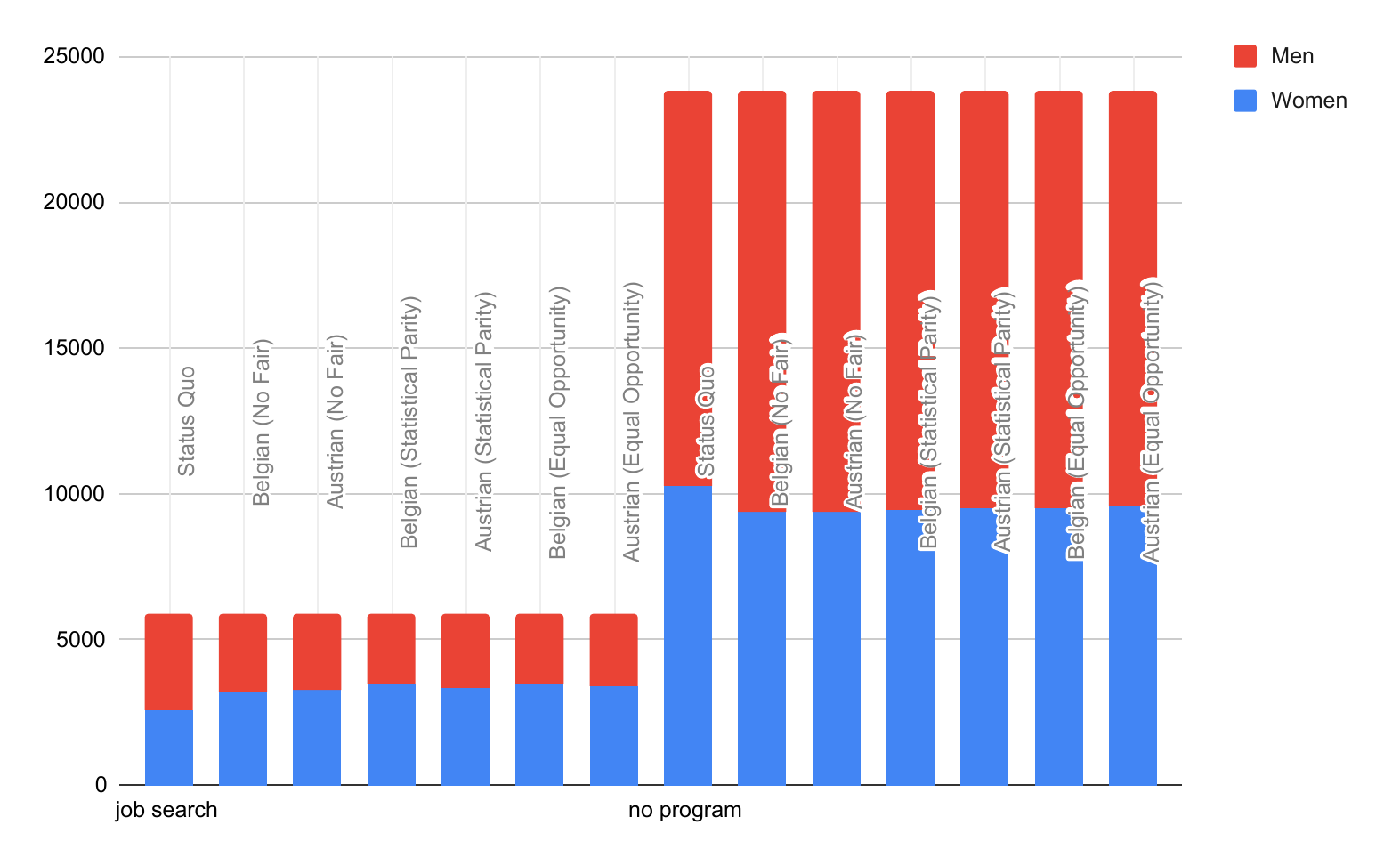}
    \end{subfigure}
    \caption{Program participation by gender under both algorithmically informed policies, for all risk scores, and baseline capacities. Note the different scales and, especially, the higher participation of women in vocational training and employment programs and the drop in language courses.}
    \label{fig:number-participation}
\end{figure}

\clearpage
%\subsection{Citizen LTU Gap}\label{supp:citizen}

\begin{figure*}[htb!]
  \centering

  \begin{subfigure}{0.45\textwidth}
    \includegraphics[width=\linewidth]{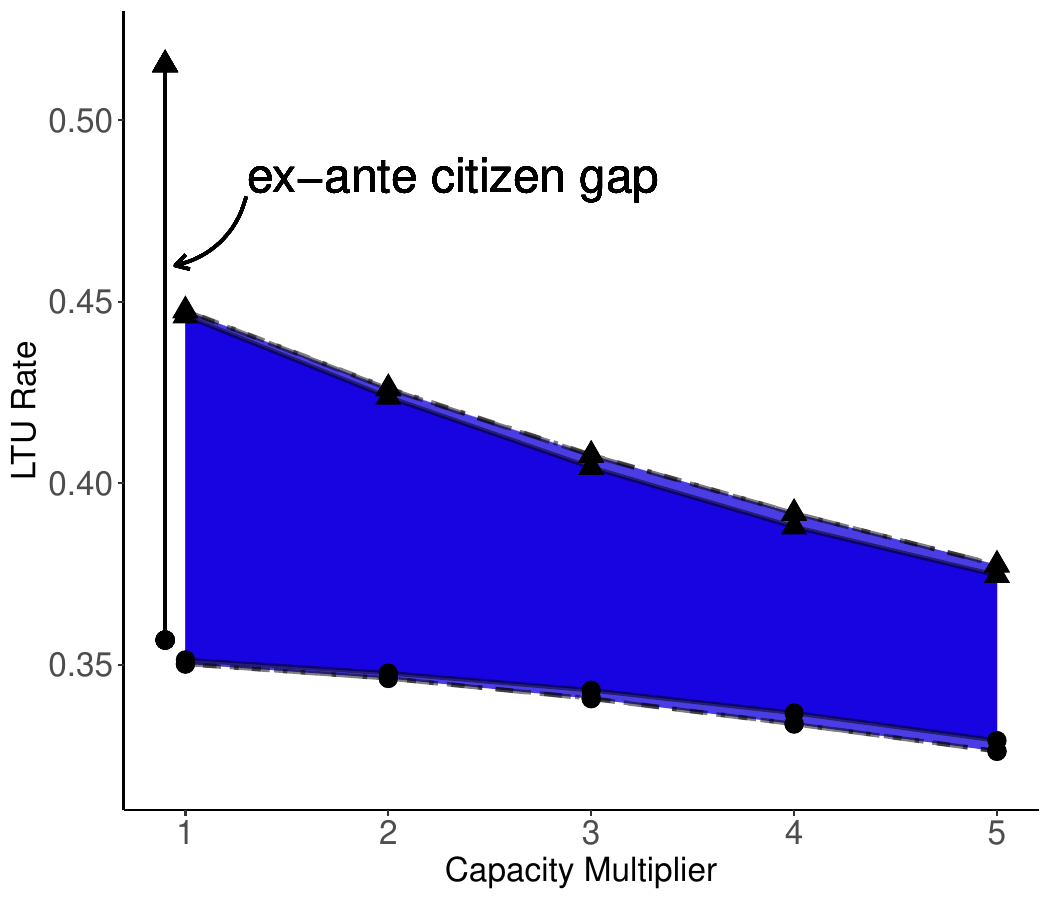}
    \caption{Belgian Prioritization and Optimal Program.}
  \end{subfigure}
  \hfill
  \begin{subfigure}{0.45\textwidth}
    \includegraphics[width=\linewidth]{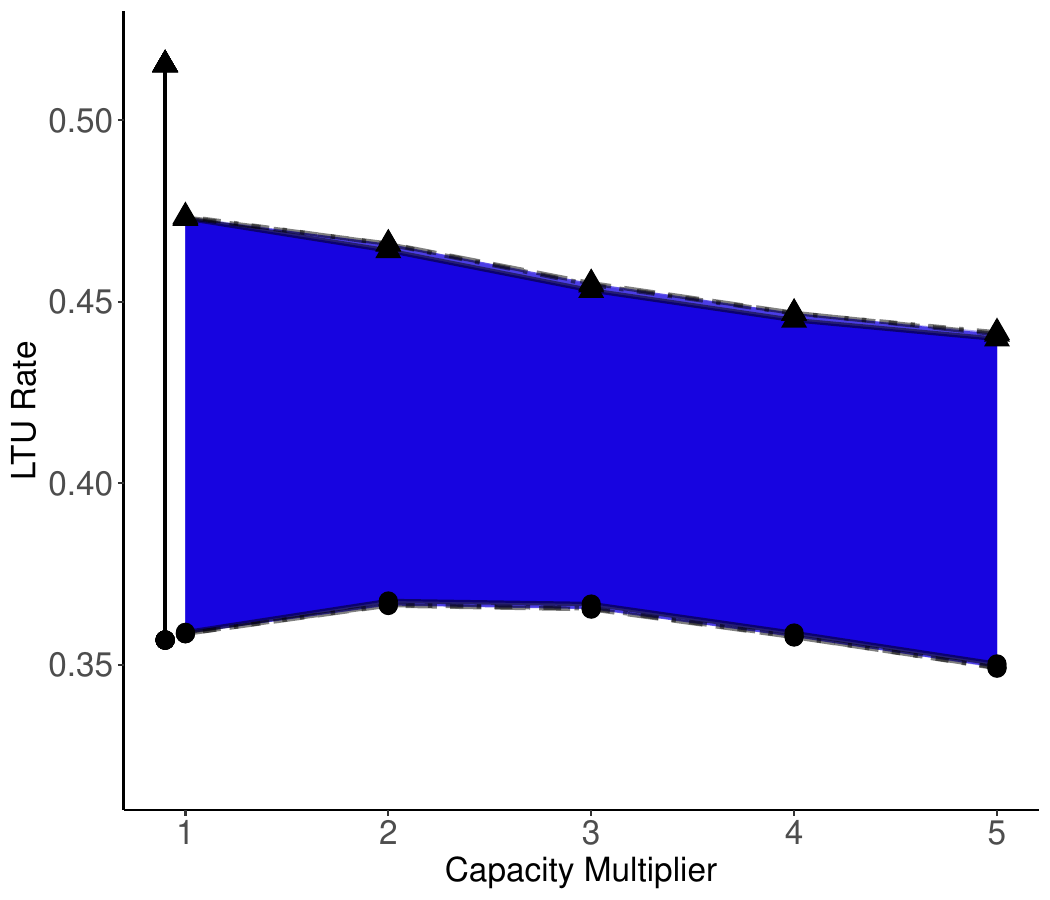}
    \caption{Belgian Prioritization and Random Program.}
  \end{subfigure}

  \medskip

  \begin{subfigure}{0.45\textwidth}
    \includegraphics[width=\linewidth]{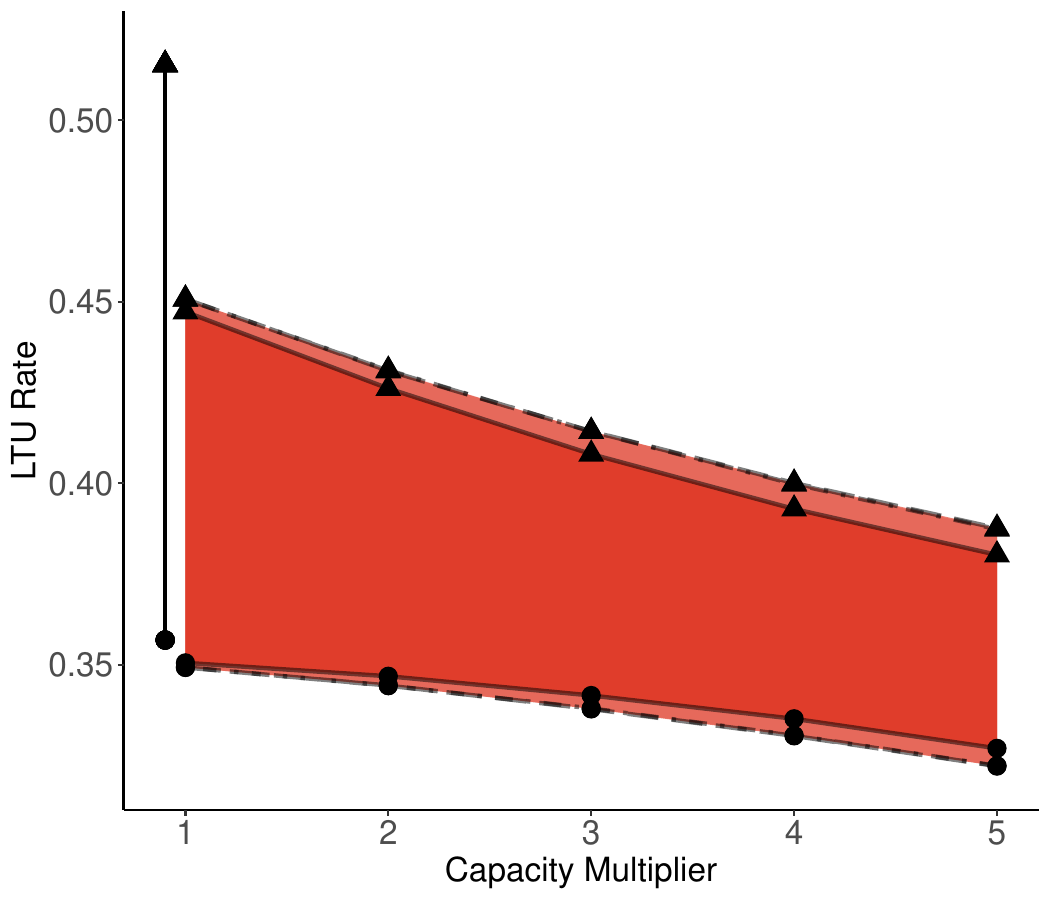}
    \caption{Austrian Prioritization and Optimal Program.}
  \end{subfigure}
  \hfill
  \begin{subfigure}{0.45\textwidth}
    \includegraphics[width=\linewidth]{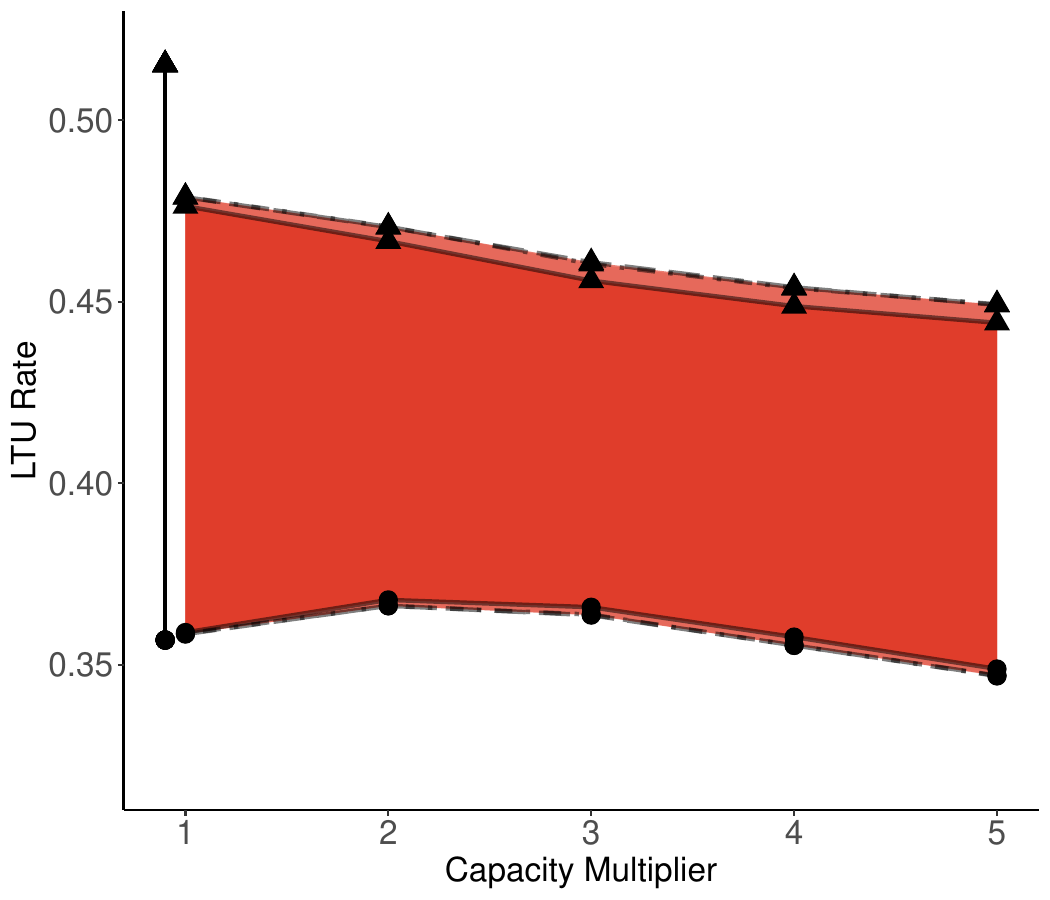}
    \caption{Austrian Prioritization and Random Program.}
  \end{subfigure}

  \medskip

  \begin{subfigure}{0.48\textwidth}
  \centering
    \includegraphics[width=0.6\linewidth]{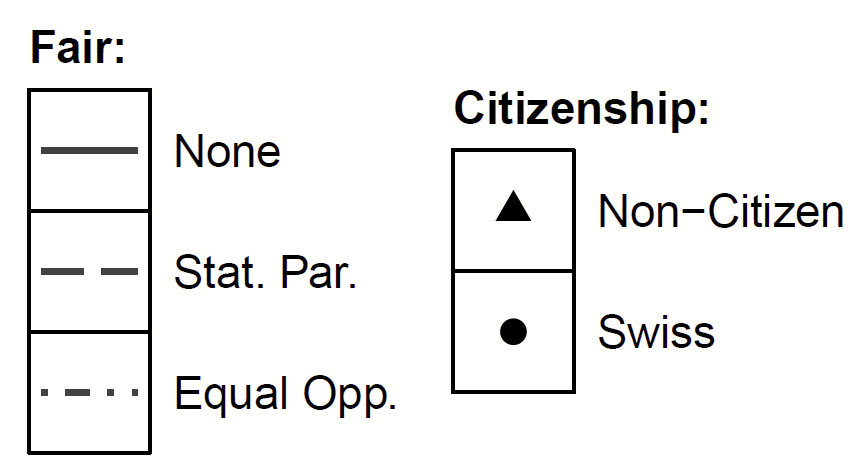}
      
  \end{subfigure}
  \caption{We plot the citizen gap in long-term unemployment (LTU) against program capacities for each combination of prioritization and assignment schemes. The level of transparency shows the citizen gap for the corresponding fairness constraint: none, statistical parity, or equal opportunity. All policy combinations reduce the citizen gap. The unconstrained risk scores (lowest transparency) result in the smallest citizen gap. This effect is especially pronounced as program capacity is increased and program assignments are individualized (optimal). Austrian prioritization compared to the Belgian approach performs particularly poorly under fairness constraints with respect to gender.}
  \label{fig:fairnessPenaltyCitizens}
\end{figure*}

\begin{figure*}[htb!]
  \centering

  \begin{subfigure}{0.48\textwidth}
    \includegraphics[width=\linewidth]{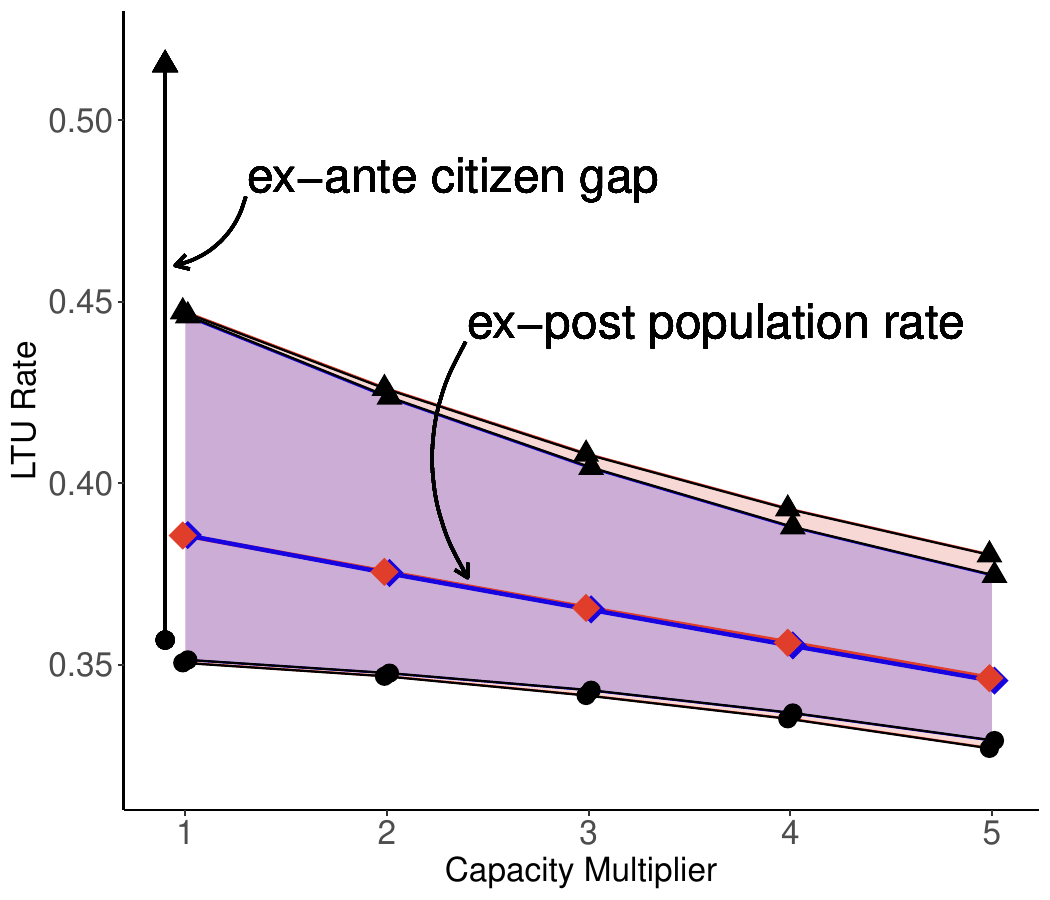}
    \caption{Optimal Program.}
  \end{subfigure}
  \hfill
  \begin{subfigure}{0.48\textwidth}
    \includegraphics[width=\linewidth]{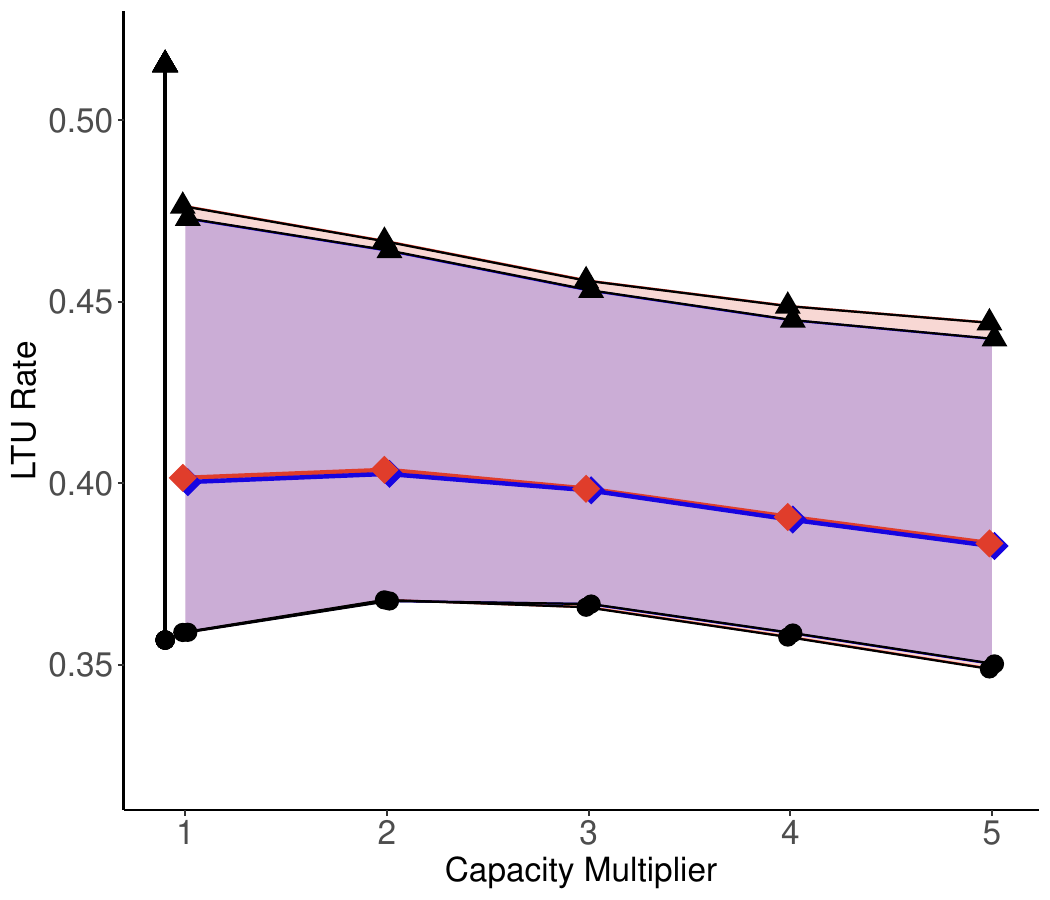}
    \caption{Random Program.}
  \end{subfigure}

  \medskip

    \begin{subfigure}{0.48\textwidth}
  \centering
    \includegraphics[width=0.6\linewidth]{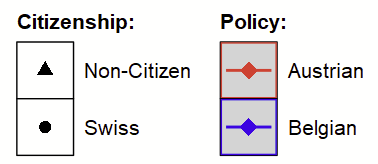}
      
  \end{subfigure}

  \caption{We plot overall long-term unemployment and the citizen reemployment gap against program capacity for each combination of prioritization and assignment scheme. For clarity, results are shown only for fairness-unconstrained risk scores. Regardless of the assignment scheme, the Belgian prioritization (blue line) results in the same long-term unemployment rate as the Austrian and a slightly smaller citizen gap. Individualized program assignments (optimal) are markedly more effective, especially under larger program capacities.}
  \label{fig:CitizenOverall}
\end{figure*}

%\clearpage
%\subsection{Gender LTU Gaps for (un)married (non-)citizen}\label{supp:married-swiss}

\begin{figure*}[htb!]
  \centering

  \begin{subfigure}{0.48\textwidth}
    \includegraphics[width=\linewidth]{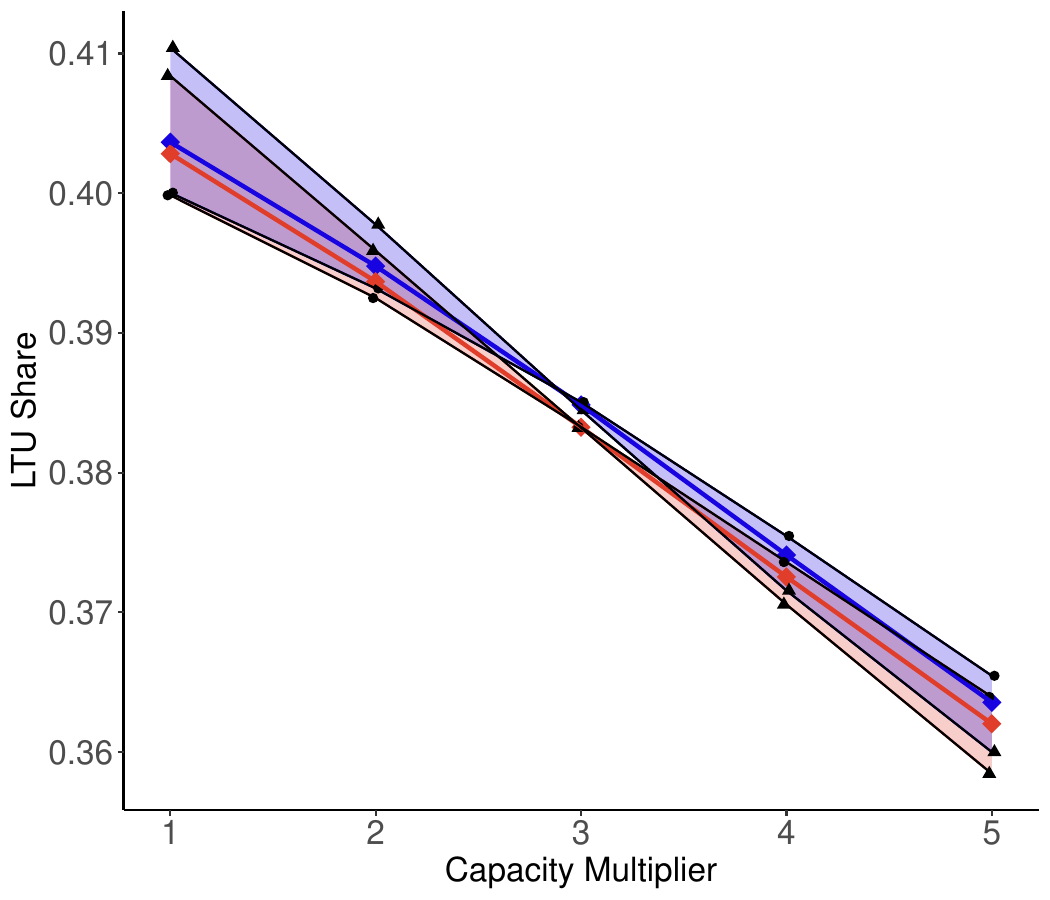}
    \caption{Unmarried Non-Citizen.}
  \end{subfigure}
  \hfill
  \begin{subfigure}{0.48\textwidth}
    \includegraphics[width=\linewidth]{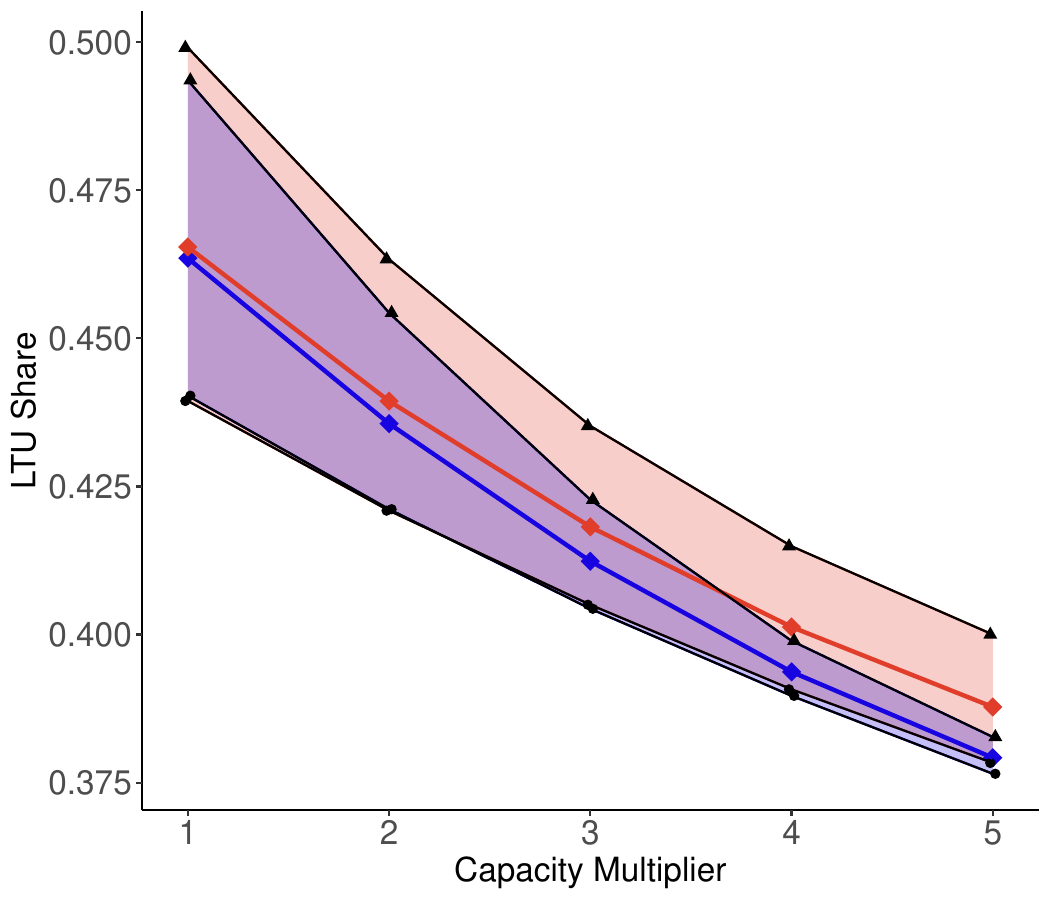}
    \caption{Married Non-Citizen.}
  \end{subfigure}

  \medskip

  \begin{subfigure}{0.48\textwidth}
    \includegraphics[width=\linewidth]{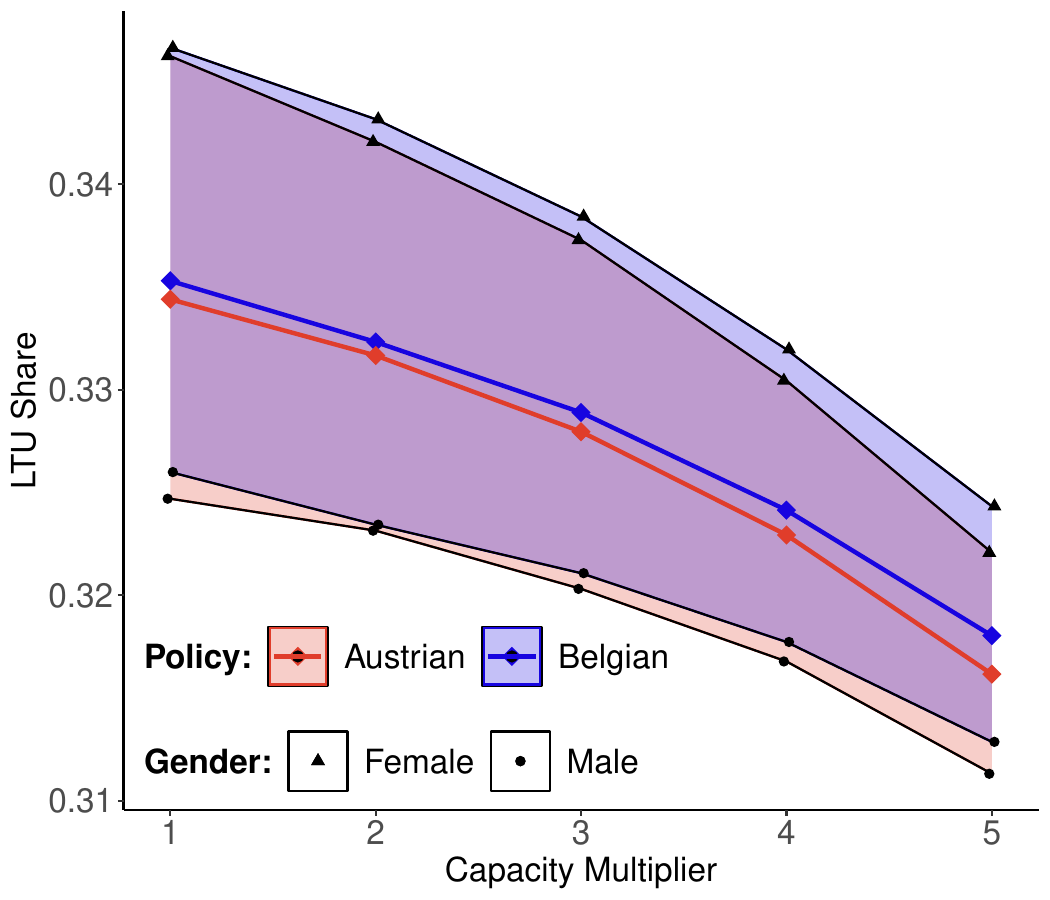}
    \caption{Unmarried Swiss Citizen.}
  \end{subfigure}
  \hfill
  \begin{subfigure}{0.48\textwidth}
    \includegraphics[width=\linewidth]{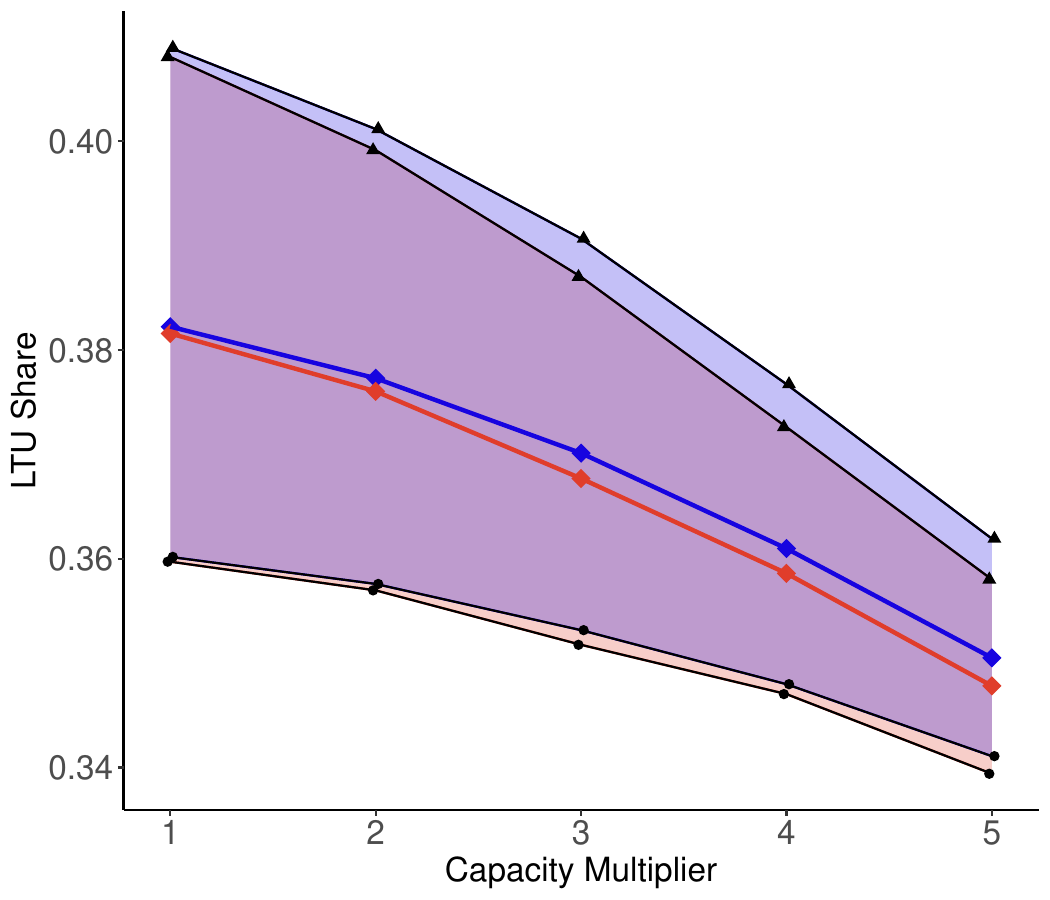}
    \caption{Married Swiss Citizen.}
  \end{subfigure}

    \caption{We show the overall long-term unemployment (LTU) rates by prioritization scheme (red and blue line) and by gender for four sub-groups: unmarried non-citizen, unmarried Swiss citizen, married non-citizen, and married Swiss citizen. All results are based on fairness unconstrained risk scores for LTU and optimal assignment. Note the different scales. The reduction in LTU rates and the gender gap is especially pronounced for the group of married foreigners. For unmarried foreigners, the gender gap even flips under both algorithmic policies at four- and five-fold program capacities.}
    \label{fig:married-swiss}
\end{figure*}

%\clearpage
%\subsection{Risk scores and optimal potential outcomes}

\begin{figure*}[htb!]
    \centering

    \begin{subfigure}[b]{0.49\textwidth}
        \centering
        \includegraphics[width=\textwidth]{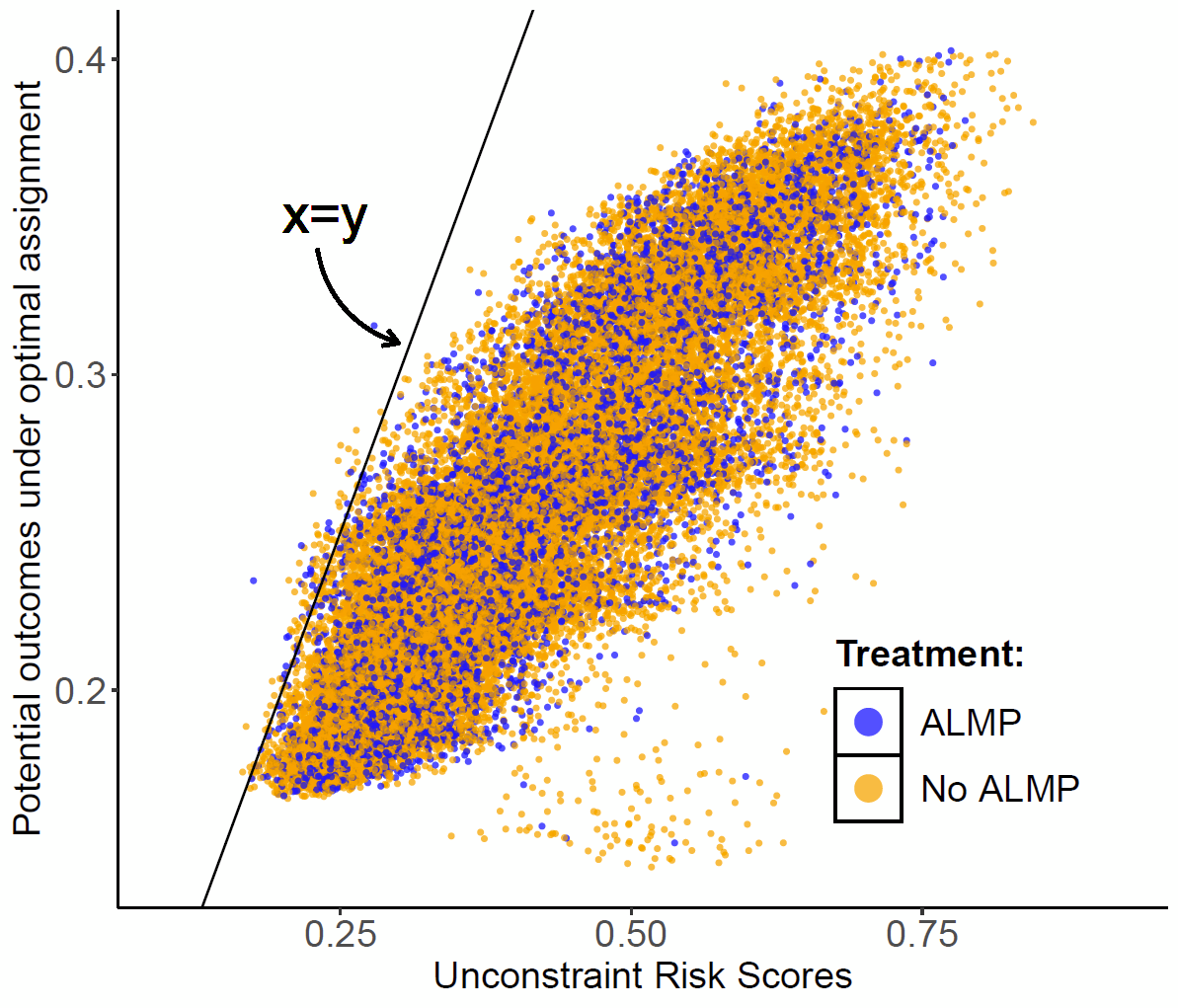}
        \caption{Risk scores without any fairness constraint plotted against the optimal (minimal) potential outcome. Spearman's rank correlation is $\rho=0.864.$}
        \label{fig:a-risk-vs-po}
    \end{subfigure}
    \hfill
    \begin{subfigure}[b]{0.49\textwidth}
        \centering
        \includegraphics[width=\textwidth]{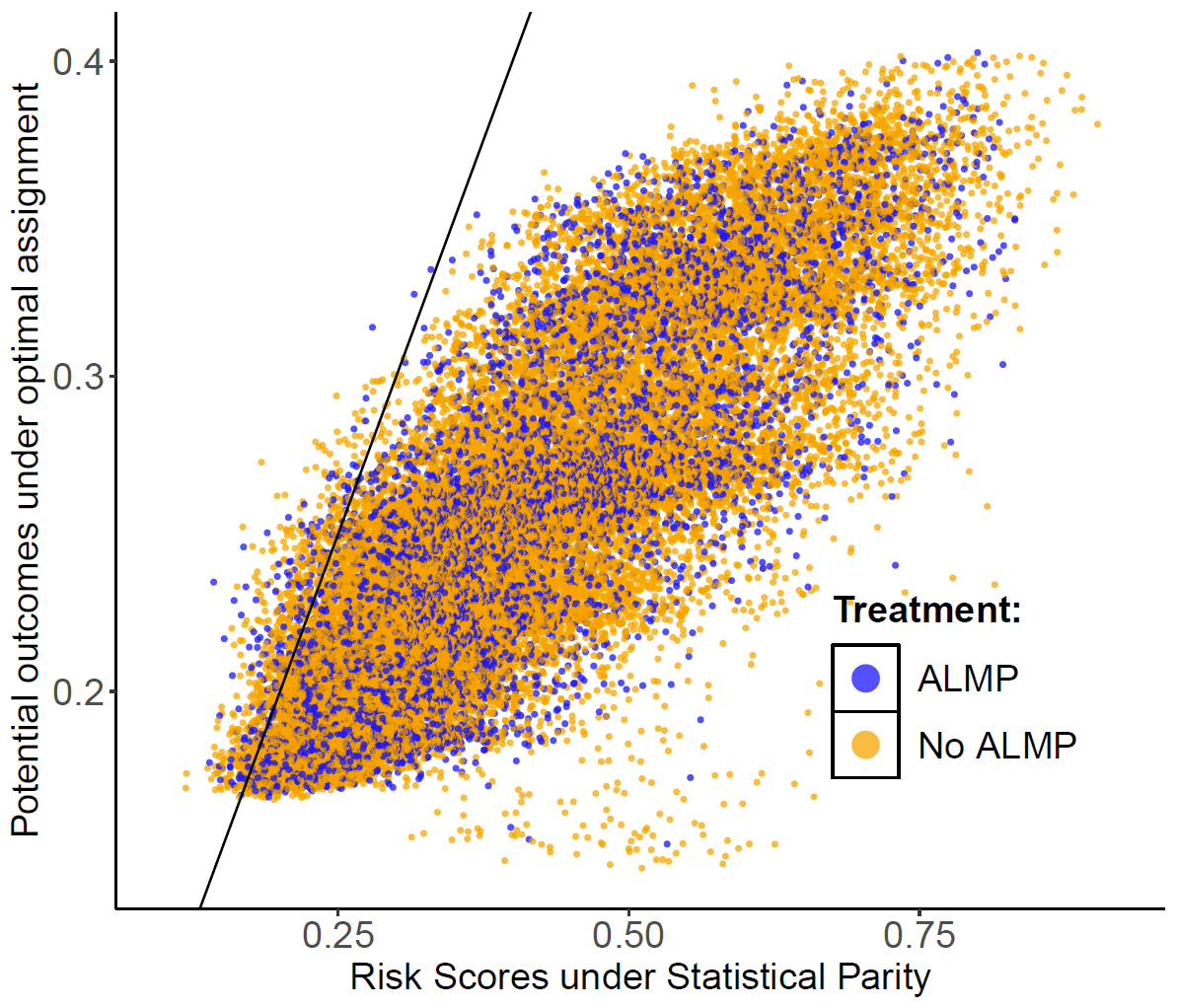}
        \caption{Risk scores with the statistical parity constraint plotted against the optimal (minimal) potential outcome. Spearman's rank correlation is $\rho=0.832.$}
        \label{fig:b-risk-vs-po}
    \end{subfigure}
    \hfill
    \begin{subfigure}[b]{0.49\textwidth}
        \centering
        \includegraphics[width=\textwidth]{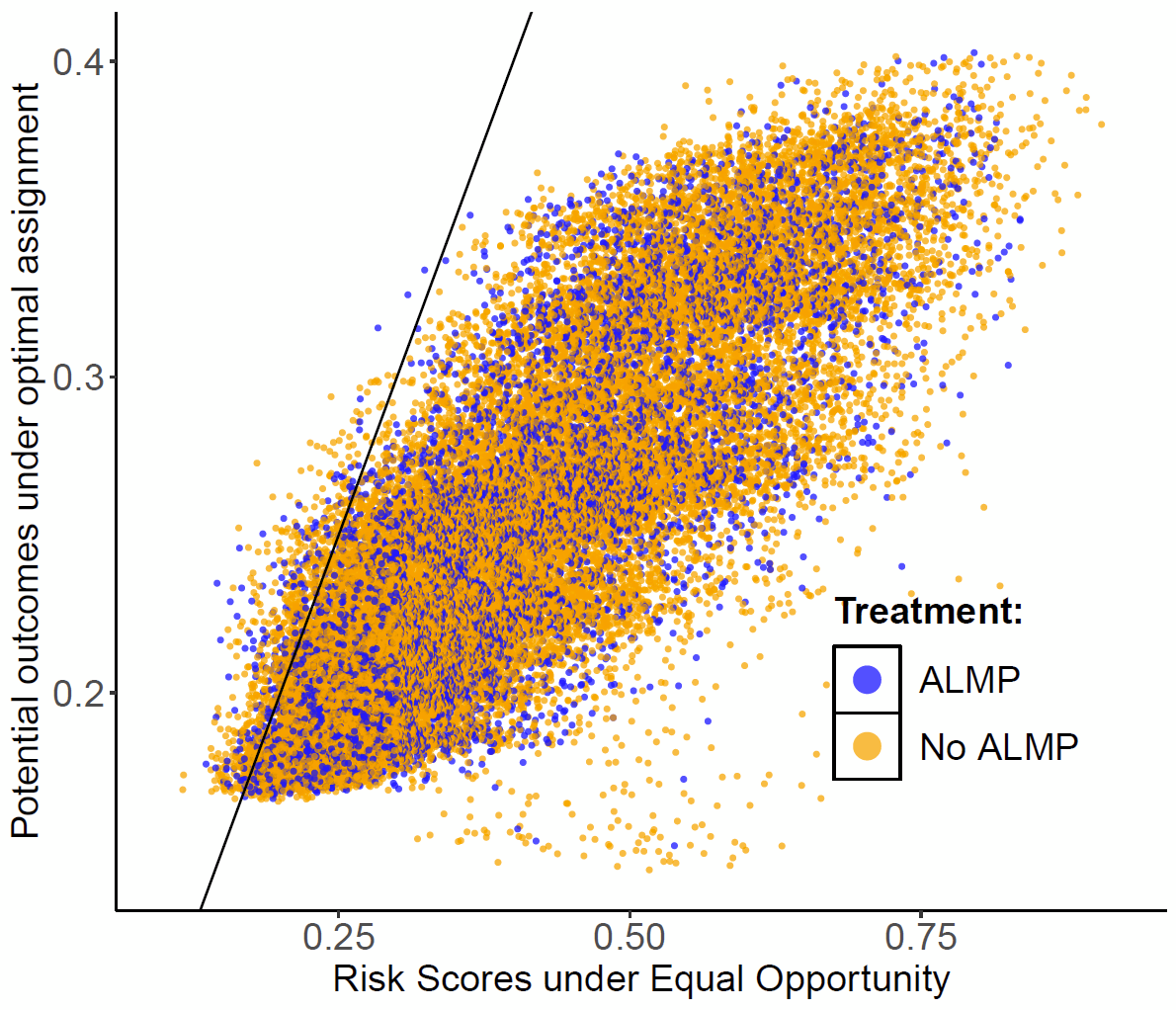}
        \caption{Risk scores with the equal opportunity constraint plotted against the optimal (minimal) potential outcome. Spearman's rank correlation is $\rho=0.826.$}
        \label{fig:c-risk-vs-po}
    \end{subfigure}
    
    \caption{We plot the respective (fairness constraint) risk scores for long-term unemployment (LTU) against the estimated individually optimal (minimal) potential outcomes. All three risk scores are biased estimates of the optimal potential outcome. An unbiased estimate would scatter around the diagonal line shown. The fairness constraints additionally increase the variance.}
    \label{fig:risk-vs-po}
\end{figure*}

\end{document}